\documentclass[12pt,reqno]{amsart}

\usepackage{graphicx}
\usepackage{mathabx}
\usepackage{bbold}
\usepackage{mathrsfs}               
\usepackage{amssymb}
\usepackage{graphicx, amsmath}
\usepackage{amsfonts}
\usepackage{dsfont}
\usepackage{amssymb}
\usepackage{fancyhdr}
\usepackage{a4wide}
\usepackage{xcolor}
\usepackage{enumerate} 
\usepackage[colorlinks]{hyperref}
\usepackage{mathtools}
\usepackage{scalerel}
\usepackage{bm}
\usepackage{braket}
\usepackage{enumerate}
\usepackage{enumitem}
\usepackage{xfrac}				
\usepackage{nicefrac}				

\usepackage{enumitem}

\usepackage[colorlinks]{hyperref}
\usepackage{mathtools}

\numberwithin{equation}{section}

\newcommand{\bV}{\boldsymbol{V}}
\newcommand{\bE}{\boldsymbol{E}}
\newcommand{\bF}{\boldsymbol{F}}
\newcommand{\blambda}{\boldsymbol{\lambda}}
\newcommand{\bx}{ {\bf x}}
\newcommand{\bp}{{\bf p }}
\newcommand{\be}{ {  \bf e } }
\newcommand{\ba }{  \textbf{\textit{a}}     }

\newcommand{\bA}{ {\bf A}  }
\newcommand{\mast}{ \overline{\mu}_\ast  }
\renewcommand{\d}{\mathrm{d}}

\newcommand{\N}{\mathbb{N}} 
\newcommand{\C}{\mathbb{C}} 
\newcommand{\Z}{\mathbb{Z}} 
\newcommand{\R}{\mathbb{R}} 
\newcommand{\w}{w}
\newcommand{\U}{\mathcal{U}}
\newcommand{\W}{\mathcal{W}}
\newcommand{\calD}{\mathcal{D}}

\newcommand{\calF}{\mathcal{F}}

\newcommand{\F}{\mathcal{F}}
\renewcommand{\S}{\mathbb{S}}
\renewcommand{\H}{\mathcal{H}}
\newcommand{\calC}{\mathcal{C}}
\newcommand{\calU}{\mathcal{U}}

\newcommand{\im}{\mathrm{Im}}

\newcommand{\A}{ {\bf A}  }
\newcommand{\sg}{  \boldsymbol{\sigma } }
\newcommand{\p}{\mathbb{P}}

\newcommand{\delto}{\delta_0}
\newcommand{\bchi}{{\boldsymbol \chi }	}

\newcommand{\auxE}{\widetilde{E}}

\renewcommand{\im}{\mathrm{i}}
\newcommand{\1}{ \mathds{1} }

\newcommand{\Cr}{ \mathscr{C}_{ \mathrm{r} } }

\newcommand{\supp}{\textnormal{supp}}
\newcommand{\<}{\left\langle}
\renewcommand{\>}{\right\rangle}
\newcommand{\vp}{\varphi}
\newcommand{\ve}{\varepsilon}
\renewcommand{\t}[1]{\textnormal{#1}}
\newcommand{\core}{\mathcal{C}}

\newcommand{\dmu}{\mu_\ast}

\newtheorem{condition}{Condition}

\newtheorem{conditionalt}{Condition}[condition]

\newenvironment{conditionp}[1]{
	\renewcommand\theconditionalt{#1}
	\conditionalt
}{\endconditionalt}

\newtheorem{theorem}{Theorem}[section]
\newtheorem{definition}{Definition}
\newtheorem{proposition}[theorem]{Proposition}
\newtheorem{corollary}{Corollary}
\newtheorem{lemma}[theorem]{Lemma}

\theoremstyle{remark}
\newtheorem{notation}[theorem]{Notation}

\theoremstyle{definition}
\newtheorem{remark}[theorem]{Remark}

\newcounter{listi}
\newenvironment{remarklist}{\begin{list}{{\rm(\roman{listi})}}{%
		\setlength{\topsep}{0mm}\setlength{\parsep}{1mm}\setlength{\itemsep}{0mm}%
		\setlength{\labelwidth}{1.3em}\setlength{\leftmargin}{1.2em}\usecounter{listi}%
}}{\end{list}}

\renewcommand{\le}{\leqslant}
\renewcommand{\leq}{\leqslant}
\renewcommand{\ge}{\geqslant}
\renewcommand{\geq}{\geqslant}     

\title[Tunneling estimates for magnetic Dirac systems]{Tunneling  estimates for two-dimensional  perturbed magnetic Dirac systems}

\author[E. C\'ardenas]{Esteban C\'ardenas}
\address[E. C\'ardenas]{Department of Mathematics,
University of Texas at Austin,
2515 Speedway,
Austin TX, 78712, USA}
\email{eacardenas@utexas.edu}

\author[B. Pavez]{Benjam\'in Pavez}
\address[B. Pavez]{Instituto de F\'isica, Pontificia Universidad Cat\'olica de Chile, Vicu\~na Mackenna 4860, Santiago 7820436, Chile}
\email{bipavez@uc.cl}

\author[E. Stockmeyer]{Edgardo Stockmeyer}
\address[E. Stockmeyer]{Instituto de F\'isica, Pontificia Universidad Cat\'olica de Chile, Vicu\~na Mackenna 4860, Santiago 7820436, Chile}
\email{stock@fis.puc.cl}

\begin{document}

\maketitle

{\centering\footnotesize 
	\textit{Dedicated to the memory of Georgi Raikov.}\par}

\begin{abstract}
	We prove tunneling estimates for  two-dimensional Dirac systems which are localized in space due to the presence of a magnetic field.
	The Hamiltonian driving the motion  admits the decomposition $  H  = H_0 + W$, where $H_0 $ is a rotationally symmetric magnetic Dirac operator and  $W  $  is a position-dependent  matrix-valued potential satisfying certain smoothness condition in the angular variable. 
	A consequence of our results are  upper bounds for  the 
	growth in time 
  of  the expected size of the system and its total angular momentum.
\end{abstract}

\setcounter{tocdepth}{1}
{\hypersetup{linkcolor=black}
	\tableofcontents}

\section{Introduction}
\label{section introduction}
Consider a quantum  system whose dynamics is driven by a self-adjoint Hamiltonian $H_0$ on the 
 Hilbert space $L^2(\R^2, \C^q)$ for $ q\in \{1,2\}$. 
Assume that $H_0$ is symmetric under rotations so that it can be decomposed in an orthogonal sum  of operators acting on  the corresponding  angular momentum subspaces \emph{i.e.} 
 	\begin{equation*}
 H_0 \cong\bigoplus_{j\in \Z} h_j\,.
 \end{equation*}	
 Let us  assume  further that the $h_j$'s have compact resolvent, \emph{i.e.} discrete spectrum. In particular, the unperturbed system  has certain  localization properties in space which may or not be uniform in $j$.
 We are interested in perturbing such systems and analyse the fate of certain dynamical variables that describe the size of the quantum system.
 
 Let us consider    a hermitian multiplication operator $W$, 
 possibly not invariant under rotations, that it is smooth in the angular variable and it is such that  the operator sum
 \begin{align*}
 H=H_0+W\,,
 \end{align*}
   is a well-defined   self-adjoint operator. 
   According to the Schr\"odinger equation,
    if $\varphi\in L^2(\R^2, \C^q)$ is the initial  state of the system at time 
   $t=0$,  then  (in units where $\hbar=1$)   
    \begin{align*}
   \varphi(t)=e^{-{\rm i}t H }\,\varphi\,,
   \end{align*}
   is the quantum state at time $t>0$.
   Let $	\p_I (H)$ denote the spectral projection of $H$ onto a 
   {bounded} 
   energy  interval $I\subset \R$. We are motivated by the study of the long time behavior of the expected radius of the system for finite energy initial data.
   More precisely, the problem we are interested in consists
   of understanding 
   	\begin{align}\label{eq.size}
   R_\varphi(t) : =
   	\<  \vp(t)  ,   |\bx|  \,    \vp (t)  \>\,,\quad \varphi=\p_I (H)\varphi\,, \quad t\gg 1\,.
   \end{align}
In certain cases, 
 \emph{apriori} knowledge on the spectrum of $H$    
 is helpful. 
 For instance,  when the spectrum   is purely absolutely continuous  one expects ballistic dynamics  {\emph{i.e.}}  $R_\varphi(t)\sim t$ holds \cite{Last1996}. Another  case is   when the spectrum is a discrete set, where one can show that the dynamics is localized and, in particular, that $\sup_{t>0} R_\varphi(t) <+\infty$ holds, at least for various physically relevant models. 
 The former cases correspond to when the spectral measure of $H$
 is  absolutely continuous with respect to the Lebesgue measure and 
 a point measure, respectively.
 However, generically, the  spectral measure may  behave 
 very differently  from the former two cases and in such situations  there is very little to say about \eqref{eq.size}. 
Moreover,   in view of the Weyl - von Neumann and the Wonderland Theorems (see \emph{e.g.} \cite{Oliveira2009} and \cite{RJLS2,hundertmark2008short}),
the spectral quality of the system may be very unstable under perturbations. 
 Therefore, it is  desirable to  study the underlying dynamics of such systems without relying on   the spectral quality of $H$.  
 Indeed, notice that    even if all   the  $h_j$ have  compact resolvent,
  it is unclear whether  or not the same holds for 
  $H_0$. 
 Hence, the problem on the behavior of $R_\varphi(t)$ is non-trivial even for $W=0$.

 An approach to address this problem
 for the magnetic Schr\"odinger operator
 was proposed  by Hundertmark, Vugalter and two of the present authors in  \cite{CHSV21}. 
 In there, $H_0$ is the   Schr\"odinger operator on $L^2(\R^2, \C)$  with 
 rotationally symmetric magnetic field 
\begin{equation}
\label{B field}
 \textbf{B} (\bx )=     		  B (	 |  \bx  	|	) 	 \be_3 
 	\qquad
 	 \t{with}
 \qquad 
 B( r ) \sim r^{\alpha -1 } \t{ as } r \rightarrow \infty \ . 
\end{equation}
 pointing orthogonal to the plane, for some $\alpha>0$. 
 The perturbation $W$ is an electric potential, 
assumed to be 
decaying  at infinity and to be   infinitesimally 
 $-\Delta
 $-bounded. 
 The  results of the present paper cover 
 the situations  considered in 
  \cite{CHSV21} analogous to  the magnetic Dirac operator on $L^2(\R^2, \C^2)$.
   
 The strategy in \cite{CHSV21} is based upon two steps. Firstly, one shows certain tunneling estimates, describing  the decay in the probability for the system to be localized in the classically forbidden  region. 
 This allows in  particular to estimate  
 $R_\varphi(t)$  in terms of the 
 expectation value 
   of  some power of 
   the angular momentum operator $   - i \partial_\theta$.
 Secondly, one estimates the time behavior  of the latter expectation using Heisenberg's evolution equation and again the tunneling  
 estimates. It turns out that the faster $W$ decays the better are the bounds on $R_\varphi(t)$.
 
 Let us now turn to the model investigated in this work. Consider from now on $H_0$ to be  the Dirac operator on $L^2(\R^2, \C^2)$ with 
  magnetic field \eqref{B field}.
Contrary to the  non-relativistic case, a growing \textit{electric potential} enhances the delocalization properties of a Dirac system. This can already be seen  in the simple case when $H_0$ is perturbed by a (well-behaved) rotationally symmetric electric potential $V$ (\emph{i.e.}  $W=V\otimes 		\textbf{1}_{2\times 2}$). 
 Let us further explain. 
To this end, let
 \begin{align}
 A(r) = \frac{1}{r}\int_0^r s B(s) \d s\, 
 \end{align}
be the corresponding magnetic vector potential. 
Then,   it is known \cite{KMSYamada} that  the corresponding
operator  $h_j+V$ has  compact resolvent if and only if 
 	\begin{align}\label{eq.V/A}
\mu_* := 
\limsup_{|\bx|\to\infty} |V(|\bx|)/A(|\bx|)|^2 
 <1\,.
 \end{align}
 Since $A\sim |\bx|^\alpha$ the electric potential $V$ is allowed to grow at infinity, but slower than $A$. 
 In fact, if 
 $\mu_* =1 $
 and $V>0$, the spectrum of $h_j+V$ splits
 as discrete on the positive and absolutely continuous on the  negative real line \cite{Yamada,KMSYamada}. Moreover, if $ \mu_\ast> 1$ the spectrum of $h_j+V$ is purely absolutely continuous \cite{KMSYamada} and behaves for large times in time average as $R_\varphi(t)\sim t$ \cite{MSBallistic2015}. Hence, the system has no localization property.
  
  In this work we consider the Dirac system $H=H_0+W$
  where $W$ is an hermitian  $2\times 2$-matrix valued multiplication operator that is  (see Condition~\ref{condition 2}) locally in $L^2$, it is smooth in the angular variable, and it is allowed to grow at infinity as long as it is controlled by  $A$. 
  We quantify the angular regularity of $W$ 
  with a parameter $\beta>0$.
  In particular, 
  our condition on $W$
  	is asymptotically sharp in the limit of maximal regularity.
  	More precisely, 
  for $\beta \rightarrow \infty$ one recovers \eqref{eq.V/A}.  
  See Remark~(\ref{remark1.2})
  for more details. 
  
   Our main result, Theorem~\ref{theorem tunneling estimates}, provides the above mentioned tunneling estimates for $H$.  
   They state that for finite energy  the system remains localized  in a certain 
   	classically allowed region. 
The system has a different classically allowed region for each angular momentum channel $ j \in \Z $, 
determined by $h_j$
and a reference energy level $E_j$.
We denote by 
 $\mathfrak{C} = \big( \calC_j (E_j)	\big)_{ j \in \Z}$
 the collection of all  these sets. 
Mathematically,
it induces  an 
\textit{orthogonal projection}
$ {   \boldsymbol{\theta}	} = \oplus_{ j \in \Z } \theta_j $
in the Hilbert space $L^2(\R^2,\C^2). $
   In particular, Theorem~\ref{theorem tunneling estimates} shows      exponential decay 
 away from the support of 
    $  {   \boldsymbol{\theta}	}    $ (see Corollary~\ref{cor.p}).
  An important  consequence  of  Theorem~\ref{theorem tunneling estimates} is that for any $\nu>0$ one finds constants $c, C_1, C_2>0$ such that for any $\varphi=\p_I (H)\varphi$
   	\begin{align}\label{eq.ib2}
    C_1\<  \vp  ,   |J|^{\nu/(1+\alpha)} \,    \vp   \> -c\|\vp\|^2\le \<  \vp  ,   |\bx|^\nu  \,    \vp   \> \le  
    C_2\<  \vp  ,   |J|^{\nu/(1+\alpha)} \,    \vp   \> +c\|\vp\|^2
   \, , 
   \end{align}
where $J  = - i \partial_\theta +\frac{1}{2}\sigma_3$
is the total angular momentum. 
   As for the dynamical applications, notice first that since the spectral projection commutes with time evolution, we get  for all $t>0$
  	\begin{align}\label{eq.ib}
    R_\varphi(t)  \le  
  C\big(\|\vp\|^2+\<  \vp(t)  ,   |J|^{1/(1+\alpha)}  \,    \vp (t)  \>\big)
   \,,\quad \varphi=\p_I (H)\varphi\,.
  \end{align}
  Equations \eqref{eq.ib2} and \eqref{eq.ib} combined give a satisfactory answer to the problem when  $W$   is rotationally symmetric. 
  Indeed,   in this case the time evolution operator commutes with 
  the angular momentum operator. 
  Consequently, we deduce 
  that  $\sup_{t\ge0}R_\varphi(t)<+\infty $ holds, whenever $R_\varphi(0)$ is finite.
  If the symmetry is perturbed the situation is  less clear and the problem deserves further investigation.
  However, in \cite{CHSV21}  a general tendency
  was shown. 
  Namely,  if  $W$ is smooth and
  decays at infinity,    
  it is proven that 
  $R_\varphi(t)\lesssim t^\epsilon$
  where 
 the exponent $\epsilon>0$ depends on the decay rate of $W$.
 In particular, if  $W$ decays sufficiently fast, 
 one gets  sub-ballistic dynamics \emph{i.e.} $\epsilon\in(0,1)$.
  Moreover, for exponentially decaying $W$ one gets that $R_\varphi(t)$  grows at most 
logarithmically 
   in time.  These results can be replicated here using  Theorem~\ref{theorem tunneling estimates} as an essential ingredient.  
%
%

  Let us now describe the new elements in our analysis. 
  While our  strategy  for the proof of Theorem \ref{theorem tunneling estimates}  follows  the  general lines of \cite{CHSV21}, 
   there are two new  major  independent  obstacles:
   \begin{enumerate}[leftmargin=*]
   	\item The Dirac operator $H$ is not semi-bounded.
Let us recall that the heart of the proof of \cite{CHSV21}
consists of introducing a comparison operator $\widetilde H$
that creates a spectral gap around zero of its  Schr\"odinger counterpart. This is achieved by 
adding a suitable  operator that \textit{shifts} the spectrum to the right. 
Since $H$ is not lower bounded, 
the  same  naive operation   fails and no spectral gap is created. 
   	We overcome this by introducing a novel comparison operator,   formally given by  
   	\begin{align}\label{co}
   		\widetilde{H} = H+ \sigma_3\bE \bchi \ . 
   	\end{align}
   	The operator 
   	$
   	 \sigma_3 \bE \bchi 	
   	 =	
   	  \oplus_{ j \in \Z } \, \sigma_3  E_j\chi_j
   	 $
   	acts a mass term that  depends on the angular momentum $j\in\Z$ under consideration.
In particular, it    	  creates a suitable spectral gap around zero in the sense of Proposition \ref{prop inverse estimate}. 
The sequence   	$(E_j)_{j\in \Z}$ corresponds to a   collection  of
carefully selected 
   	\textit{reference energy levels},  prescribed in terms 
   	{ of $ W$ and  the  energy interval $I$. 
   	For each $ j \in \Z$, the energy level $E_j$ determines  
   	an allowed region $\calC_j (E_j)\subset \R_+$ 
   	 on which each $\chi_j \sim \1_{\calC_j(E_j)}$
   	acts as a smooth   {localization function}. 
 This is a novel adaptation of the comparison operators suggested for Schr\"odinger-like operators in \cite{Griesemer2004} and may be of interest in  its own right, even for constant $\bE$.}   	
   	\item 
   	The perturbation $W$ is allowed to grow at infinity. 
   	In particular, $W$ may not be relatively bounded with respect to $H_0$.
   	We recover control of
   	the perturbation
   	by carefully analyzing its growth when projected onto the 
   	{classically allowed region $\mathfrak{C}$}. 
   	In particular,  we prove  in Theorem~\ref{theorem W} that  there  are constants  $a\in (0,1)$, $\mast\in [\mu_\ast,1)$, and $c>0$ such that
   	\begin{align}\label{eq.W-bound}
   		\|W\varphi\|^2\le a\|H_0\varphi\|^2+\mast\|{\boldsymbol{\lambda \chi}}\varphi\|^2
   		+ c\|\varphi\|^2 \ , 
   	\end{align} 
for all $\varphi$ on a suitable core. 
   {The term in the middle quantifies the possible growth of $W$ at infinity:
the operator $\bchi = \oplus_{ j \in \Z } \chi_j$
is the same as in \eqref{co}
and   $\blambda = \oplus_{ j \in \Z } \lambda_j$ 
is a sequence of multipliers
such that  $ \lambda_j \sim A(|j|^{	1/ (1+ \alpha )	})$ for large $|j|$; see  Definition \ref{definition lambda}.
 We select the energy levels such that $\mast{\boldsymbol \lambda}^2< {\bE}$.
	In particular,  if $W$ 
	decays at infinity,  
	 then $\mast=0$ and we recover the usual relative operator bound inequality. 
 For this    case,  the $E_j$'s can be selected to be all equal to each other; 
	 see Definition~\ref{definition-energies}.   
   	This corresponds to the case considered in \cite{CHSV21} with only one energy level. }
   \end{enumerate}

 Summarizing, in order to prove Theorem \ref{theorem tunneling estimates}
we  show   exponential decay of $\p_I (H)$ 
 on the support of $(1-{\boldsymbol \chi})$.
  We proceed in the spirit of \cite{BFS1998b} using an integral representation of the spectral projection. 
 The method works efficiently 
 in terms of the comparison operator $\widetilde H$ which is a well-defined self-adjoint operator thanks to \eqref{eq.W-bound} and
 satisfies 
  $\p_I(\widetilde{H})=0$. 
 \vspace{3mm}

 {\it Organization of the paper.} 
 In Section \ref{section main results} 
 we define  the precise model that we study,
state our main results and discuss its consequences. 
 In Section \ref{section preliminaries} we collect some  preliminary facts  used throughout  our analysis. 
 The analysis of the classically allowed region and the definition of the corresponding localization  function ${\boldsymbol \chi}$ is given in  Section \ref{section analysis car}.
 This is then used in  Section \ref{section perturbation bounds}
  to obtain bounds on $W$ as in \eqref{eq.W-bound}; see Theorem \ref{theorem W}.
 In Section \ref{section abs decay thm} we 
 define the energy levels and  state
  a somewhat  abstract exponential  decay result, Theorem~\ref{theorem exp decay}. Moreover, we show how that result implies Theorem~\ref{theorem tunneling estimates}.  
  Finally, in Section~\ref{section proof thm} we show  Theorem~\ref{theorem exp decay}. Most of the  section is devoted to establish the main properties of the comparison operator $\widetilde{H}$. The main text of this article is followed by several appendices containing useful technical results.

\section{The model and main results}
\label{section main results}
In this section  we define the mathematical model that we study. We then state our main result in Theorem \ref{theorem tunneling estimates} and discuss its consequences. 
\subsection{The model}
We consider a massless Dirac system in a two-dimensional plane, under the influence of  a
 magnetic field pointing orthogonal to the plane
$      \textbf{B}(\bx ) = B(\bx )\be_3$, with $\be_3=(0,0,1)^{\rm T}$.  The Hamiltonian describing this system operates on the 
Hilbert space 
$\H : = L^2(\R^2, \C^2)$ and may be defined on  $C_c^\infty (\R^2  , \C^2 )$ acting as
\begin{equation}\label{h-0}
	H_0   = {\sg   } \cdot (-\im \nabla - {\bf A} )\,,
	\end{equation}
	where 
	$ ( \sg  , \sigma_3 )  \equiv  (\sigma_1, \sigma_2, \sigma_3 )$  correspond to   the standard  Pauli matrices
	\[\sigma_1=\begin{pmatrix}
	0&1\\
	1&0
	\end{pmatrix}\,,\quad \sigma_2=\begin{pmatrix}
	0&-i\\
	i&0
	\end{pmatrix}\,,\quad \sigma_3=\begin{pmatrix}
	1&0\\
	0&-1
	\end{pmatrix} \,,\]
	and  $\bA =(A_1,A_2)^{{\rm T}}: \R^2 \rightarrow \R^2$ is a magnetic vector potential
	satisfying $\partial_1 A_2-\partial_2 A_1=B$. If $B$ is    rotationally symmetric we fix   the gauge to 
	\begin{equation}\label{aaa}
		\A (\bx) = A( |\bx| )   \, \textbf{e}_\theta 
		\qquad \t{ and } \qquad 
		A(r) = \frac{1}{r}\int_0^r s B(s) \d s\,.
	\end{equation}
	(The vectors  $\textbf{e}_\theta $ and $\textbf{e}_r $ 
	are the usual orthonormal basis of $\R^3$ in polar coordinates with $\textbf{e}_r\times \textbf{e}_\theta=\textbf{e}_3 $).
	Notice that in this  gauge the magnetic vector potential is related to the magnetic  flux $\Phi$ on the centered ball $B_r(0)$ of radius $r>0$ 
	given by
	\begin{align}\label{flux}
	\Phi(r)=\int_{B_r(0)} \mathbf{B} \cdot {\mathrm{d} \bx}\,,
	\end{align} 
	through $\Phi(r)=2\pi rA(r)$.
	
	In order to exploit the symmetry we identify $\H$ with $ L^2(\R_+) \otimes L^2(\S^1)^2$, where   $L^2 (\R_+)\equiv L^2 (\R_+,r\d r)$  is equipped with the measure $r\d r$. We define  
	the \textit{total angular momentum operator}
	\begin{equation}
		J  = 		\textbf{1}
		\otimes \left(- \im \partial_\theta  
		+
		\frac{1}{2}
		\sigma_3 \right)
		\quad 
		\t{on}\quad L^2(\R_+) \otimes L^2(\S^1)^2,
	\end{equation}
	which is self-adjoint in its maximal domain. We denote by
	$p_j = (2\pi)^{-1}|e^{ij\theta}\rangle\langle e^{ij\theta}|$  
	the projection  from  $L^2(\S^1)$ to the space of orbital angular momentum $j\in \Z$. 
	The resolution of the identity  for $J$ is given by 
	\begin{align}
\label{J}
	J
	=
	\sum_{j\in \Z} 
	m_j 
	P_{m_j}
	\,,
	\qquad
	\t{where}
	\qquad
	m_j := j  +  \frac{1}{2}   
\end{align} 
and 	where
	\begin{equation}
		P_{m_j} 
		=
		\textbf{1}
		\otimes
		\begin{pmatrix}
			p_{j} & 0 \\
			0 &  p_{j+1}\\
		\end{pmatrix}
		\,.
	\end{equation}
	The sequence of projections $ \big(P_{m_j}\big)_{ j \in \Z }$ provide a   decomposition
	 of the Hilbert space 
	 $\H\cong \bigoplus_{j\in\Z} L^2(\R_+, \C^2)$ in angular momentum subspaces.
	Since $H_0$ commutes with $J$ \cite[Chapter ~7]{Thaller1992}, the action of the Hamiltonian $H_0$   is described on each angular momentum subspace as
	\begin{equation}\label{Dj definition}
		H_0 \cong\bigoplus_{j\in \Z} h_j\,, \qquad \mbox{where} \qquad h_j :=  
		\begin{pmatrix}
			0 &    d_j^* \\
			d_j  & 0 
		\end{pmatrix}\,,
	\end{equation}	
	and 
	\begin{align}
		&d_j := \frac{d}{d r}  - \frac{j}{r} + A(r)\, ,  \quad d_j ^*:= -\frac{d}{d r}  - \frac{(j+1)}{r} + A(r) \,. 
	\end{align}

 Let us now state the assumptions of the magnetic field of the unperturbed problem. They are given in terms of the flux function \eqref{flux} as follows.

	\begin{condition}\label{condition 1} 
	The magnetic field   $B\in L_{\rm loc}^2(\R^2, \R )$ is 
	  rotationally symmetric.
	  Further, there exist constants $\alpha>0$, $\Phi_0>0$
	  and a  function $\phi$, {with $\phi(r) = o (r^{1+ \alpha })$ as $r\to \infty$},
	  such that 
		\begin{align}
			\label{magnetic field}
			\Phi ( r ) 
			\, =  \, 
			\Phi_0 \, 
			r^{ \alpha +1 }   
			+ 
			\phi(r)  \ , \qquad r> 0 \ . 
		\end{align}		
	\end{condition}
	\begin{remark}
The operator $H_0$ is essentially self-adjoint on 
$C_c^\infty(\R^2,\C^2)$  and 
we denote its closure by $H_0$ as well. 		
The following can be said about its spectrum (see \emph{e.g.} \cite{Thaller1992}):
		\begin{enumerate}[label=(\roman*)]
			\item If $\alpha  > 1 $, ${\rm spec}(H_0)$   is purely  discrete. 
			\item  If $\alpha = 1 $, ${\rm spec}(H_0)$  is a perturbation of Landau levels. 
			\item if $ 0<\alpha <1$, ${\rm spec}(H_0)$ coincides with $\R$, is purely essential, and dense pure point. 
		\end{enumerate}   
	\end{remark}
	We shall consider  perturbations  $W$ of $H_0$ that have  general hermitian matrix structure given through multiplication by 
	\begin{align}\label{eq-for-W}
	 W(\bx)=  \begin{pmatrix}
			\w_{11}(\bx) &    \w_{12}(\bx) \\
			\w_{21}(\bx)  & \w_{22}(\bx) 
		\end{pmatrix}\, , \qquad \bx \in \R^2 \ . 
	\end{align}
	\begin{remark}\label{rem.generalw}
		Recall that the set $\{  \mathds{1}_{2\times 2}, \sigma_1 , \sigma_2 , \sigma_3   \}$ form a  basis for the space of  two-dimensional hermitian matrices. 
		Hence, the perturbation $W$ may be written in the form
		\begin{equation} 
			W (\bx )  \ = \    \,  V  ( \bx )\otimes 
			{\textbf{1}}_{2\times 2}     \ + \         \sg \cdot  \ba (\bx )  \ + \  \sigma_3 m(\bx )  
			\ , \qquad \bx \in \R^2 
		\end{equation} 
		for some functions $ V : \R^2 \rightarrow \R $, $ \ba :   \R^2 \rightarrow \R^2 $ and $ m : \R^2 \rightarrow \R $. 
		From  a physical point of view,  $V$ corresponds to an electric potential, $\ba$ corresponds to a magnetic vector potential and $m$ corresponds to a position-dependent mass term also known as Lorentz scalar potential. 
	\end{remark}

 Let us now state the  precise conditions 
that we impose on the perturbation $W$.
To this end, 	let us first introduce 
the matrix-valued function on $\R_+\times\Z $
 \begin{equation}
 		\widetilde{W}(r,j) : =\begin{pmatrix}
 		\widehat{w}_{11}(r,j) &    \widehat{w}_{12}(r,j-1) \\
 		\widehat{w}_{21}(r,j+1)  & \widehat{w}_{22}(r,j) 
 	\end{pmatrix} \ , 
 \end{equation}
where  for 
  $w \in L^1_{\t{loc}}(\R^2, \C)$
  we write   its
  Fourier coefficients  as   
$
  \widehat{w}(r,j) : =    \frac{1}{\sqrt{2 \pi }} \int_0^{2 \pi} w (r \textbf{e}_r )  e^{ - \im j \theta } \d \theta  .
  $
 	Essentially, the function     just defined  is the Fourier transform   of the map   $\theta\mapsto {W}(r,\theta)$ 
	and the additional $\pm1$   account for      the spin  angular momentum, 
	introduced in \eqref{J}. 
In particular, we note that for $\vp , \psi \in \H $ smooth enough there holds
	\begin{equation}
		\< \vp, W \psi \>_{\H}
\ 		 =  \ 
		 \sum_{j,k} 
		 \,
		 \langle  \vp_j,  \widetilde{ W}(\cdot, j -k ) \psi_k   \rangle_{	 L^2(\R_+ ,\C^2 )	} , 
	\end{equation}
where $\vp_j  \equiv  P_{m_j} \vp $ and  $\psi_k \equiv  P_{m_k} \psi $
are regarded as vectors in $L^2( \R_+ , \C^2	)$.

\begin{notation}
	Given a matrix $A\in \C^2 \times \C^2$, 
$\|	A\|_{2 \times 2 }$ denotes 
	its operator norm in $\C^2$ with respect to
	the usual inner product $ (x,y)= \sum_{i=1}^2 \overline{x_i} y_i$. 
	Hence, $|(x,Ay)| \leq \| A	\|_{2\times2} |x| |y|$. 
\end{notation}

	\begin{condition}\label{condition 2}
 $W \in L^2_{\t{loc}} (\R^2, \C^2 \times \C^2 )		$ 
 is a hermitian matrix-valued multiplication operator.
		Further, there exist    $\beta>0$, $\mu_\ast\in[0,1)$ 
		  and   
		  $v \in L_{\t{loc}}^2	( [0,\infty )  ,\R	)	 $ 
		such that:
		\begin{enumerate}[label=(\alph*)]
			\item\label{smoothness} (Smoothness in $\theta$). For all $r>0$ and $j\in \Z$ we have
			$$
\| \widetilde{W}(r,j)\|_{2\times 2} 
			\le v(r) e^{-\beta | j| }
			\,.
			$$ 
			\item\label{a-control} (Control by $A$ at infinity). 
	\begin{align}\label{eq.lim}
	{\limsup _{r \to \infty}} \,\,\,\left|\frac{\coth^2(\beta/4) v^2(r)}{A^2(r)}\right|=\dmu<1\,.
\end{align}
		\end{enumerate}
	\end{condition}
	\begin{remark}\label{rem.rem}
		\begin{remarklist}
			\item In view  of Paley--Wiener Theorem,  Condition \ref{condition 2} \ref{smoothness} corresponds to uniform real-analiticity. It is reminiscent from 
			the works on Gaussian-like decay for eigenfunctions of the perturbed   Landau Hamiltonian initiated in \cite{laszlo}  and reinterpreted in \cite{Shu-gaussian}. As noted in \cite{CHSV21} it can be improved to Gevrey regularity, in that case the constant $\coth(\beta/4)$ should be replaced by one given in terms of the Gevrey scale.  
			\item \label{remark1.2}

			The case $\beta =\infty$ is allowed 
			and corresponds to the limit of maximal regularity.
			That is,  when 
$W$ is   rotationally symmetric. 
	In this situation,  
	the condition \ref{a-control} is sharp: 
 for  $\mu_*  \geq 1  $ absolutely continuous spectrum may appear, 
 and 
	  no localization properties can be expected (see   the discussion 
	  in the Introduction). 
			
		\end{remarklist}
	\end{remark}
	In Appendix~\ref{appendix local} (see Remark~\ref{rem.ess.sa}) we show that under  the stated conditions on the fields $B$ and $W$  
	\begin{align}\label{h-w}
		H=H_0+W\,,
	\end{align}   
	is essentially self-adjoint on  $C_c^\infty (\R^2  , \C^2 )$.
	We denote its closure by the  same symbol  and its domain by  $\calD (H)$. 
	\begin{remark}
		The spectrum of $H$ undergoes through several regimes under the stated conditions. In particular, if $W$ is an electric potential that grows at infinity  it is known that the resolvent of $H$ is compact if  
		$
		\limsup_{|{\bf x}|\to \infty} W^2/(2B)<1 
		$.
		However, if this inequality  does not hold, essential spectrum may appear (see \cite{MS2012}). (Notice that $B\ll A^2$ for large $r$ and compare with \eqref{eq.lim}.) 
		In particular,     if $  V^2/B \to \infty $ as $r\to \infty$ one can construct examples where the spectrum of $H_0+V$ covers the whole real line with point spectrum  \cite[Example 1]{MS2012}. 
	  We believe  that if the symmetry of this system is 
	  slightly perturbed singular continuous spectrum should appear.  
		\end{remark}
\subsection{Main result}
	\begin{notation}
		We often write for the inverse of the flux exponent 
		$$\sigma=\frac{1}{1+\alpha}\,.$$
	\end{notation}

	\begin{notation}
	Throughout this work we denote   by $\1_{I}$ the indicator  function  on a set $I\subset \R$.  For the  spectral projection onto energies in  $I$ we write
	\begin{equation}
		\p_I (H) \equiv \1_{I}  { (H)} \ . 
	\end{equation}
Identity functions are denoted by $\textbf{1}_A$
for a relevant space $A$. 
\end{notation}
	
	Our main result below  states that the probability of finding a particle with bounded energy and angular momentum $m_j$, away from certain region that scales  like $|j|^{\sigma}$, decays exponentially fast. 
	The  decay rate $\sim e^{-\zeta  r^{\alpha+1}}$  is ruled by the behavior of the flux function \eqref{flux}. For $x\in \R$, we denote by $\< x \> \equiv (1 + x^2)^{1/2}$ the standard  bracket.

	\begin{theorem}
		\label{theorem tunneling estimates}
Let   $H$  be the Dirac Hamiltonian,    satisfying  Conditions  \ref{condition 1} and \ref{condition 2}, 
and let $E>0$ and $I=[-E,E]$. 
Then, there exist positive constants 
 $\zeta_1$, $\zeta_2$, and $C_1<C_2$ such that  
		\begin{enumerate}[label=(\roman*)]
			\item The Interior Tunelling Estimate holds 
			\begin{equation}
				\sum_{ j \in \Z } 
				\, 
				\|
				\,
				\exp \big(    \zeta_1    |j | \big) 
				\, 
				\1_{ [ 0, C_1 |j|^\sigma]   }(r )
				P_{m_j}
				\p_I( H )
				\, 
				\|^2 
				< \infty    \ \label{equation interior tunnelling}. 
			\end{equation} 
			\item The Exterior Tunelling Estimate holds 
			\begin{equation}
				\sum_{ j \in \Z }  
				\ 
				\|
				\, 
				\exp \big(   \zeta_2  \,   r^{ 1 + \alpha }  \big) 
				\1_{   [ C_2 \<j \>^\sigma  , \infty )}(r ) 
				P_{m_j}
				\p_I (   H ) 
				\, 
				\|^2 
				< \infty     \ \label{equation exterior tunnelling}. 
			\end{equation}
		\end{enumerate}
	\end{theorem}
	\begin{remark}
		\begin{remarklist}
			\item The norms in \eqref{equation interior tunnelling} and \eqref{equation exterior tunnelling} correspond to the operator norm from $\mathcal{H}$ to $L^2(\R_+,\C^2)$.
			\item By taking $W={\boldsymbol \sigma}\cdot {\bf a}$, the same result extends to the Pauli operator with magnetic vector potential ${\bf A}+{\bf a}$ which is defined through the quadratic form $$
			q(u)=\|(H_0+{\boldsymbol \sigma}\cdot {\bf a})u\|^2\,,\quad u\in \mathcal{D}(H)\,.
			$$
			\item This result and its dependency on the flux power (here $\alpha+1$) is formally the same as in  
			\cite[Theorem 3.1]{CHSV21}  for the
			magnetic Schr\"odinger operator $(-i\nabla-{\bf A})^2+W$, being $W$ an electric perturbation that decays  at infinity. 	
	
			\item The \textit{speed} of the exponential decay is encoded in the parameters 
			$ \zeta_1$ and  $\zeta_2$, 
			and their dependence on 
			the magnetic field parameters $\alpha$,  $\mu_*$ and $\beta$, 
			as well  the energy $E$. 
While we do not compute explicitly this dependence, 
it follows from our proof that $\zeta_1$ and $\zeta_2$
decreases  monotonically with  $E$, and converges to zero as $\mu_* \rightarrow 1 $. 
		\end{remarklist}
	\end{remark}
\begin{remark}
	For another approach to localize Dirac particles with   magnetic
	fields see \cite{Nenciu:2021rt}. In that case the confinement on a bounded region $\Omega \subset \R^2$  is achieved by means of a magnetic field that  diverges at the boundary of $\Omega$.
	\end{remark}

{In order to prove Theorem \ref{theorem tunneling estimates} 
	it suffices to consider magnetic fields
	for which the potential $A(r)$ takes a  simpler form. 
	More precisely, 
	we consider the  following 
	alternative form of Condition \ref{condition 1} for  which $A(r)=r^\alpha$ holds. 
	\begin{conditionp}{\ref*{condition 1}$'$}
		\label{condition 1'}
 The magnetic field $B \in L^2_{	\t{loc}	}(\R^2,\R)$ is rotationally symmetric, 
 and the flux function satisfies 
 \begin{equation}
 {\Phi(r)} = 2\pi r^{\alpha+1} , \qquad r>0 \ . 
 \end{equation}
\end{conditionp}}

{
\begin{remark}\label{remark condition 1'}
In Appendix \ref{section scaling} we argue that Conditions \ref{condition 1'} and \ref{condition 2} are enough to show Theorem~\ref{theorem tunneling estimates}. In particular,  
we can  re-scale  the 
position  variables $\bx  \mapsto \lambda^{-1  } \bx  $ 
with $\lambda = ( 2\pi / \Phi_0)^{\sigma }$ and 
the effect of this operation only changes the value of the constants that appear in Theorem \ref{theorem tunneling estimates}. 
On the other hand, the tail $\phi(r)$
that is present in Condition \ref{condition 1}
can always be moved
into the potential $W$
in the form of a magnetic perturbation.
In particular, this perturbation satisfies  Condition \ref{condition 2} with $\beta = \infty$
and $\mu_* = 0$.  	
 \end{remark}}

	\subsection{Consequences of the main result}
	{
	Consider the setting of Theorem~\ref{theorem tunneling estimates}. In particular,  throughout this section, let $I=[-E,E]$, for some $E>0$, be a fixed energy interval.} 
We define the following  operator
 on $\mathcal{H}$ 
 implementing  an 
orthogonal projection 
onto the classically allowed region
		\begin{align}
	{\boldsymbol{\theta}} 
	: =
	\oplus_{j\in \Z}
	 \theta_j,\quad\mbox{where}\quad \theta_j(r) 
	 := 
	\mathds{1}_{(C_1 |j|^\sigma, C_2 \<j\>^\sigma)}(r), \quad r>0\,. 
	\end{align}

	The following result is a rather direct consequence of Theorem~\ref{theorem tunneling estimates}.

\begin{corollary}\label{cor.p}
	There are constant $\xi_1, \xi_2>0$ such that  
	\begin{align}\label{eq.eq}
(1-{\boldsymbol{\theta}} ){\rm exp}(\xi_1 |\bx|^{\alpha+1})
\p_I(H)\,
\quad\mbox{and} \quad(1-{\boldsymbol{\theta}} ){\rm exp}(\xi_2 |J|)
\p_I(H)
\end{align}
are bounded operators on $\mathcal{H}$.	
\end{corollary}
\begin{proof}
For a normalized $\varphi\in\mathcal{H}$, by 
 Parseval's identity, we have that,  for $\xi_1>0$,
\begin{equation}\label{pablo}
	\begin{split}
\big\| (1-{\boldsymbol{\theta}})
{\rm exp}(\xi_1 |\bx|^{\alpha+1})
\p_I(H)\varphi \big\|^2&=
\sum_{j\in\Z} \big\| (1-\theta_j(r))
e^{\xi_1 r^{\alpha+1}}
P_{m_j}\p_I(H)\varphi 
\big\|_{L^2(\R_+,\C^2)}^2\\
&
\le 
\sum_{j\in\Z} \big\| (1-\theta_j(r))
e^{\xi_1  r^{\alpha+1}}
P_{m_j}\p_I(H)\big\|^2 \,. 
\end{split}
\end{equation} 
Recall that $1-\theta_j=\1_{ [ 0, C_1 |j|^\sigma]   } +	\1_{   [ C_2 \<j \>^\sigma  , \infty )}$. Moreover, for  $r\in [ 0, C_1 |j|^\sigma] $
we have that 
$e^{\xi_1 r^{\alpha+1}}\le 
e^{\xi_1 C_1 |j|}$ holds. Hence, the right hand side of   \eqref{pablo} is finite thanks to 
	\eqref{equation interior tunnelling} and
	\eqref{equation exterior tunnelling} provided $0<\xi_1\le \min\{C_1^{1+\alpha}\zeta_1,\zeta_2\}$. This shows that the first operator in \eqref{eq.eq} is bounded.The  second case can be dealt with  analogously. 
 For a normalized $\varphi\in\mathcal{H}$ and a constant  $\xi_2>0$, we have	
\begin{equation}\label{pablo2}
\begin{split}
\big\| (1-{\boldsymbol{\theta}} ){\rm exp}(\xi_2 |J|)
\p_I(H)
\big\|^2&\le
\sum_{j\in\Z} \big\| (1-\theta_j(r))
e^{\xi_2|m_j|}
P_{m_j}\p_I(H)\big\|^2\\
&\le 
e^{\xi_2}\sum_{j\in\Z} \big\| (1-\theta_j(r))
e^{\xi_2|j|} P_{m_j}\p_I(H)\big\|^2 \,.
\end{split}
\end{equation}
In view of \eqref{equation interior tunnelling} we control the case when 
$r\in [ 0, C_1 |j|^\sigma]$
we pick $\xi_2\le \zeta_1$. Moreover,   for the case  $r\in [ C_2 \<j \>^\sigma  , \infty )$, since $|j|< \<j \>\le (r/C_2)^{1+\alpha}$, it is controlled if $\xi_2\le\zeta_2C_2^{1+\alpha}$ thanks to \eqref{equation exterior tunnelling}. This shows the claim for the second operator in \eqref{eq.eq} and finishes the proof.
\end{proof} 
{
	The next statement is a straightforward  application of the latter result.
	\begin{corollary}\label{cor.mon}
	Let  $f:(0,\infty)\to (0,\infty)$ be a non-decreasing  function such that, for any $\ve >0$, $\sup_{r>0}f(r) e^{-\ve {r^{\alpha+1}}}$ is finite. Then, there are constants $c,C>0$, such that, for all $\varphi\in \p_I(H) \mathcal{H}$, 
	\begin{align*}
\<\varphi, {\boldsymbol{\theta}} f(c\<J\>^{\sigma}) {\boldsymbol{\theta}}\varphi\> +\<\varphi,K_f\varphi\>\le 	\<\varphi, f(|\bx|)\varphi\>\le \<\varphi, {\boldsymbol{\theta}} f(C\<J\>^{\sigma}) {\boldsymbol{\theta}}\varphi\>+\<\varphi,K_f\varphi\>\,,
	\end{align*} 
	where $
	K_f=(1-{\boldsymbol{\theta}} )
	f(|\bx|)
	\p_I(H)
	$  is a bounded operator on $\mathcal{H}$.
	\end{corollary}
\begin{proof}
	First note that for  all $\varphi\in \p_I(H) \mathcal{H}$ we have that $ \<\varphi, f(|\bx|)\varphi\>=\<\varphi, {\boldsymbol{\theta}} f(|\bx|) {\boldsymbol{\theta}}\varphi\>+\<\varphi,K_f\varphi\>$, where $K_f$ is by Corollary~\ref{cor.p} a bounded operator. Moreover,   using the monotonicity of $f$ and the support properties of ${\boldsymbol\theta}$ we compute, denoting  $\varphi_j=P_{m_j}\varphi $,
		\begin{align*}
\<\varphi, {\boldsymbol{\theta}} f(|\bx|) {\boldsymbol{\theta}}\varphi\>&= \sum_{j\in \Z}  \<\varphi_j, \theta_j f(r) \theta_j \varphi_j\> 
 \le  \sum_{j\in \Z}  \<\varphi_j, \theta_j f(C_2 \<j\>^{\sigma}) \theta_j \varphi_j\> 
\\
& \le  \sum_{j\in \Z}  \<\varphi_j, \theta_j f(C \<m_j\>^{\sigma}) \theta_j \varphi_j\> 
= \<\varphi, {\boldsymbol{\theta}} f(C\<J\>^{\sigma}) {\boldsymbol{\theta}}\varphi\>\,,
\end{align*}
	for some positive constant $C$. This gives the upper bound. For the  lower bound one argues  analogously. This finishes the proof.
	\end{proof} 
	A special case of the above, with applications to the questions posed in the introduction,  is contained in the following result. Its proof uses only  Corollary~\ref{cor.mon} and the fact that spectral projections commute with time evolution. }
	 \begin{corollary}
	 	\label{corollary 1}
	 	For any $ \nu > 0 $  there exists $c, C> 0 $ such that for all initial states
	 	$ \vp \in  \p_I ( H  ) \H $ it holds that 
	 	\begin{equation}
	  \< \vp(t) , |\bx|^\nu  \vp(t)   \>
	 	\leq 
	 	c \<
	 	\vp(t) , 
	 	| J |^{\nu \sigma }
	 	\vp(t)
	 	\>  
	 	+
	 	C \| \vp  \|^2\ , 
	 	\qquad 
	 	\forall  t \in \R \ . 
	 	\end{equation}
	 	Here $\varphi(t)=e^{-iHt}\varphi$.
	 \end{corollary}

	 \begin{remark}
	 	It follows that,
	 	whenever  the Hamiltonian  (or, its corresponding time evolution) commutes with total angular momentum, there holds
	 	\begin{align}
	 	\< \vp(t) , |\bx|^\nu  \vp(t)   \>
	 	\leq 
	 	C
	 	\Big( 
	 	\| \vp  \|^2 
	 	+ 
	 	\<
	 	\vp, 
	 	| J |^{\nu \sigma }
	 	\vp
	 	\>
	 	\Big)   \ . 
	 	\end{align}
	 	In particular, the system is dynamically localized if the initial data $\varphi\in \p_I ( H  )  \H \cap\mathcal{D}(|J|^{\nu\sigma/2})$. This   was previously observed in \cite{Barbaroux}. 
	 \end{remark}
	As mentioned before, Theorem \ref{theorem tunneling estimates}  contains tunneling estimates for the Dirac operator which are  analogous  to those for its  
	Schr\"odinger counterpart found   in \cite[Theorem 1.3]{CHSV21}.
	In that work, it is shown --  starting from these tunneling estimates -- how  to obtain {dynamical bounds}
	for the angular momentum operators. 
	This proof can be replicated  almost without changes in our setting and we shall omit it here; we refer the reader to \cite[Section 2]{CHSV21} for   details.
	\begin{corollary}
		[Dynamical bounds for $J$]
		\label{corollary 2}
		Let $ \mu > 0$,  $     p > (\mu + 1 ) / \sigma   $, $\eta>0$ and $\lambda>0$. 
		\begin{enumerate}
			\item Assume that $ w ( \bx) = \mathcal{O} ( |\bx|^{ - p} )$ as $| \bx | \rightarrow \infty$. Then, there exists $C>0$ such that for any initial state $\vp \in  \p_I ( H  ) \H$ it holds that 
			\begin{equation}
				\<  \vp(t)  , |  J  |^\mu     \vp(t)   \> 
				\leq 
				C
				\Big( 
				\<  \vp , |  J  |^\mu     \vp  \> 
				+
				t^{  \mu \theta  }
				\|  \vp \|^2 
				\Big)  
				\ , \qquad \forall t \geq  1
			\end{equation}
			where $ \theta =  1  / ( p \sigma -1 ) \in (0,1 )$. 
			\item Assume that $ w ( \bx) = \mathcal{O} (  \exp (- \eta   |\bx |^\lambda     )   )$ as $| \bx | \rightarrow \infty$. Then, there exists $C>0$ such that for any initial state $\vp \in  \p_I  (H  )  \H$ it holds that 
			\begin{equation}
				\<  \vp(t)  , |  J  |^\mu     \vp(t)   \> 
				\leq 
				C
				\Big( 
				\<  \vp , |  J  |^\mu     \vp  \> 
				+
				\big(  \log(1 + t )  \big)^{ \mu \Theta } 
				\|  \vp \|^2 
				\Big) 
				\ , \qquad \forall t \geq  1
			\end{equation}
			where $ \Theta =  \max\{     1,     1 / (\lambda \sigma )  \} $. 
			
		\end{enumerate}
	\end{corollary}

	\begin{remark}
		\label{remark combination}
		By combining Corollaries \ref{corollary 1} and \ref{corollary 2}, one readily obtains dynamical bounds for the 
		expectation values 
	of the moments of the position operator with the same behavior as in \cite[Theorem 1.8]{CHSV21}.  
	\end{remark}

	

	%

\section{Preliminaries}
\label{section preliminaries}
In this section, we record some preliminary  facts that
we will need in the rest of this article.

\subsection{Unitary transformations and Hilbert spaces}\label{sec.unit}
In the introductory section, 
we have already mentioned the relationship 
between the rotationally symmetric Dirac operator $H_0$, 
and its block decomposition $\oplus_{j\in \Z} h_j $. 
Let us record here the unitary transformations that implement this relationship.
While we will    not display them explicitly in the main body of this work, 
they are relevant  in   the proofs of essential self-adjointess of certain  operators
provided in Appendix \ref{appendix domain core}. 

\vspace{1mm}

First, we recall that in Section \ref{section introduction}
we have introduced
the Hilbert space of polar coordinates
$L^2(\R_+,\C^2 )$, always assumed
to be equipped with   the measure $r \d r $. 
We shall be extensively working with
the 
direct sum of its copies 
\begin{equation}
 	\widehat \H \equiv \bigoplus_{j\in \Z} 
	L^2(\R_+,\C^2) \ . 
\end{equation}
Let us now further explain the relationship
between $H_0$
and its diagonal decomposition on the Hilbert space
$ 	\widehat \H $.
To this end,   we introduce a modified Fourier transform.
Namely, 
let us denote by 
$\mathcal{F}_0 : L^2( \S^1  ) \rightarrow \ell^2(\Z) $
the standard Fourier series. 
Then, we consider
the unitary transformation 
\begin{equation}\label{F unitary}
	\mathcal{F} 
	\equiv 
	 			\textbf{1} \otimes 
	\begin{pmatrix}  1 & 0 \\ 0 & i  \end{pmatrix}
	\begin{pmatrix}
		\calF_0  & 0 \\
		0 &    \calF_0 
	\end{pmatrix}
	\begin{pmatrix}
		1 & 0 \\
		0 & e^{-i\theta}
	\end{pmatrix}
	\ :  \  
	L^2(\R_+ \times \S^1  , \C^2 ) 
	\,  \rightarrow  \, 
 	\widehat \H  \ . 
\end{equation}
This map is usually found in the literature
in the form
$
(\mathcal{F}\vp )_j(r)  =
( \widehat\vp^1_j(r)   , i  \widehat\vp^2_{j+1}(r)       ) 
$
where 
$
\vp = (\vp^1, \vp^2)  \in  L^2(\R_+  \times \S^1  , \C^2   )    .
$
Secondly, we  consider  the unitary transformation that implements the change of variables from cartesian to polar coordinates in the Hilbert space at hand. 
Namely,   
\begin{equation}
	\U :  \H = L^2(\R^2, \C^2)   \longrightarrow L^2( \R_+   \times \S^1 , \C^2 )    \ , 
\end{equation}
which is  defined almost everywhere by $( \U \vp )(r,\theta)\equiv  \vp(r\cos\theta,r\sin\theta)$
for
$(r,\theta) \in \R_+ \times \S^1$. 
Thus, in terms of $\F \U : \H \rightarrow   	\widehat \H $ we have that
\begin{equation}
	(\calF \calU ) H_0 (\calF \calU )^{-1} 
	=
	\bigoplus_{j\in \Z}
	h_j  \ . 
\end{equation} 
Throughout this article,  we will often losen up the notation and omit the implemented unitary transformations. 
That is, unless confusion arises, 
we simply write 
$
\textstyle 	H_0  =  \oplus_{j\in \Z} 
h_j  
$
and similarly  $ \vp = (\vp_j)_{j\in \Z}$. 

\vspace{2mm}

\textit{Notation for the direct sum of  operators}. 
Let $\{ \mathcal O _j  \}_{j\in \Z}$  
be  a sequence of linear  operators 
acting on 
$L^2(\R_+, \C^2)$.
Unless confusion arises,  we denote  its direct sum using bold-face notation
\begin{equation}
 \boldsymbol{\mathcal O}
	 \equiv 
	 \bigoplus_{j\in \Z} \mathcal O _j \ , 
\end{equation}
always defined  in the maximal domain 
$\calD (\boldsymbol{\mathcal{O}} ) = \bigoplus _{j\in \Z } \calD ( \mathcal O _j )$.
In this notation $H_0 \cong \boldsymbol{h}$.

\subsection{Essential self-adjointness}

Let us recall that 
$C_c^\infty (\R^2  , \C^2 )$
is an operator core for both $H_0$ and $H$. 
Since we will mostly work in polar coordinates, we   find it  convenient to work with a space of smooth  functions with additional properties near $r = 0$. 
Unfortunately, contrary to the three-dimensional case, 
the space $C_c^\infty ( \R^2 \backslash \{ 0\} , \C^2 ) $ is not a domain core for $H_0$. 
Instead, we will work with a space of smooth  functions  that--together with all its derivatives--converge polinomially fast to zero, with an arbitrarily high power. 

\vspace{2mm}

This space is introduced in the following definition, 
together with a \textit{radial version}, 
and its version with an
\textit{angular momentum cut-off}.

\begin{definition}[Domain cores]
	\label{definition domain cores}
	We introduce the following spaces. 
	\begin{enumerate}[label=(\arabic*) , leftmargin=*]
		\item As 
		a subspace of 		$L^2(\R^2 , \C^2 ) $, we define 
		\begin{equation} \label{domain core}
			\mathscr{C} 
			\equiv 
			\big\{ 
			\vp \in C_c^\infty (\R^2  ,  \C^2) 
			: 
			\forall \alpha \in \N_0^2 , \  \forall N \in \N,  	| D^\alpha \vp(\bx )  | = \mathcal{O}(|\bx|^N ) \t{ as } \bx \rightarrow 0 
			\big\} \ .
		\end{equation}

		\item 
		As a subspace of 
		$L^2(\R_+ , \C^2 ) $, we define 
		\begin{equation}\label{domain core C'}
			\Cr 
			\equiv 
			C_c \big(   [0, \infty) , \C^2 \big)
			\! 	\cap  \! 
			\{  
			\psi   \in C^\infty (0,\infty) 
			: 
			\forall N \in \N, \forall k \in \N_0, \ | \psi^{(k)} (r)| = \mathcal{O} (r^N) \t{ as } r \downarrow 0 
			\}   . 
		\end{equation}  
		
		\item 
		As a subspace of $ 	\widehat \H = \oplus_{ j \in \Z } L^2( \R_+ , \C^2)$, we define 
		\begin{equation}
			\label{core C}
			\textstyle 
			\calC 
			\equiv 
			\big\{ 
			(\vp_j)_{j \in \Z}
			\in
			\bigoplus_{j\in \Z}
			\Cr 
			: 
			\exists N \in \N :  \vp_j =  0 \t{ for } |j|> N 
			\big\} .
		\end{equation}
	\end{enumerate}
\end{definition}

We shall often abuse notation and refer to any element $\psi$
of either of these spaces as a \textit{smooth function}. 

\vspace{2mm}

The most important results concerning these spaces are recorded in the following Proposition, which could be 
of independent interest to some readers. 
To the authors best knowledge, the result is new. 
For a proof, see Appendix \ref{appendix domain core} . 

\begin{proposition}[Essential self-adjointness]
	\label{prop cores}
	The following statements hold true .
	\begin{enumerate}
		\item $H_0 |_{\mathscr{C}}$ and $H |_{\mathscr{C}}$ are essentially self-adjoint. 
		\item$h_j |_{\mathscr{C }_r	}$ is essentially self-adjoint  for all $ j \in \Z$. 
		\item $ {H}_0 |_{\calC}$ 
		is  essentially self-adjoint. 
	\end{enumerate}
\end{proposition}

Unless stated otherwise, 
we   make no distinction in notation between an essentially self-adjoint operator $A $ and its self-adjoint closure  $\overline A$. $\calD (A)$
always stands for its domain of self-adjointness.

\subsection{The Pauli operator}
An essential element in our approach
is the fact that  the  square of $H_0$ is given by 
the Pauli operator. 
That is (at least formally) the following holds true  
\begin{equation}
	H_0^2  = 
	(-i\nabla-\bA )^2+\sigma_3 B \,  , 
\end{equation}
where $\sigma_3$ is the third Pauli matrix. 
Of course, the Pauli operator is rotationally symmetric and can be decomposed into blocks
of constant angular momentum
$H_0^2 = \oplus_{j\in\Z} h_j^2$,
where the operators $h_j$ were
 first introduced  in  \eqref{Dj definition} (see \emph{e.g.} \cite{Thaller1992}). 
In particular, these blocks
can be put in the  convenient form, 
which we shall be using the rest of the article.
More precisely, it holds true in the sense of quadratic forms on
$\mathscr{C}_r$
that 
\begin{equation}
	\label{hj^2}
	h_j^2
	=
	T+V_j 
	\qquad
	\t{ where }
	\qquad 
	V_j
	\equiv  \begin{pmatrix}
		V_j^+&    0 \\
		0  &  V_j^-
	\end{pmatrix}   
\quad
\t{and}
\quad 
T \equiv  -  \bigg( \frac{\d^2 }{\d^2 r }+ \frac{1}{r} \frac{ \d }{ \d r }\bigg)  \ . 
\end{equation}
Here, the operator    
$
T\geq 0 
$
corresponds to the 
radial \textit{kinetic energy},  and 
\begin{equation}
	\label{effective potentials}
	V_j^\pm (r) :=  \Big(  \frac{m_j \pm 1/2}{r} - A(r)       \Big)^2 \pm B(r) 
\end{equation}	 
corresponds to     the  {\it effective potentials}. 

It will be crucial in our techniques to study
the  spatial   regions in which a 
classical particle could be found, if moving under  the influence of the effective potentials. 
More precisely, we work extensively with the following collection of  sets.

\begin{definition}
\label{definition CAR}
	Let  $ j \in \Z$.  
	Then,
	given $E_j > 0 $, 
	we define the 
	following set 
		as the 
	classically \textbf{ allowed region}  (AR) of angular momentum $j$ 
	\begin{equation}
		\mathcal{C}_j(E_j)
		\equiv
		\left\{ r > 0\,:\, V_j (r) \le E_j^2 \right\}\,.   
	\end{equation}
	Similarly, we define its complement as the 
classically		\textbf{forbidden region} (FR) 
	\begin{equation}
		\mathcal{C}_j(E_j)^\perp  \equiv ( 0,\infty) \backslash  \mathcal{C}_j(E_j) \,.
	\end{equation}
	We refer to $E_j$ as the reference energy level of angular momentum $ j \in \Z$. 
	For the full classically allowed region we write
\begin{equation}
	 \mathfrak{C}
	  = 
	  \Big(
	  \calC_j	(E_j	)
	  \Big)_{ j \in \Z } \ . 
\end{equation} 
\end{definition}

\begin{remark}\label{rem.grie}
%
%
%
%
%
%
%
	Strictly speaking
	the classical character of these regions  is directly associated to the Hamiltonian $h_j^2$.
	In this sense,   one may     understand $\mathcal{C}_j$ as the {\it energetically} allowed region of  the system with  Hamiltonian    $h_j$,  with energies  in the range $[-E_j, E_j]$.	
	
	\begin{remark}
In the above definition,  the inequality $V_j(r) \leq E_j^2 $ is understood in the sense of quadratic forms on $\C^2$. 
Since the effective potential is a diagonal matrix, 
this is  equivalent to  the inequalities for its components 
    $V_j^\pm 	(r) \leq E_j^2 $. 
	\end{remark}

\end{remark}

\section{Analysis of the clasically allowed regions }
\label{section analysis car}
In this section, 
we    analyze the clasically allowed regions given in Definition \ref{definition CAR}.
In particular, in view of Remark \ref{remark condition 1'}, 
we consider only $A(r) = r^\alpha$. 
This analysis consists of two parts:

First,  		in Subsection \ref{bcfr} we characterize the scale in which  the classically allowed regions can be found.
 Namely,   
  if the energy levels grow at most like  $ E_j   = \mu  |j |^{2\alpha \sigma }$
  for some $\mu \in(0,1)$ as $ |j |\rightarrow \infty$, 
 we show that  
 \begin{equation}
\label{localization of the CAR}
 	\calC_j (E_j) \subset [	  \delta |j|^\sigma , c  |j|^\sigma 		]
 	\qquad
 	\t{where}
 	\qquad 
 	\sigma = \frac{1}{1 + \alpha} \  , 
 \end{equation}
for large enough $|j|$. 
 Here,  $ \delta $ and 
 $c $ are appropiately chosen positive parameters, 
 independent of $ j \in \Z$. 
 We refer to \eqref{localization of the CAR} as the \textit{localization} of the 
 classically allowed regions.
 This is accomplished by  proving lower bounds for the effective potentials $V_j$ when localized away from $\mathcal{C}_j(E_j)$  (see Lemma \ref{lemma lower bound V}).

 Second, in Subsection  \ref{partitions} 
 we introduce 
 a smooth version of the projections
 $\1_{	[\delta  |j|^\sigma ,  c|j|^\sigma ]		}$, 
 which should be understood as  smooth localization functions around the ARs. 
More precisely, 
for each $j\in \Z$  we consider  
 a pair $(\chi_j , \chi_j^\perp)$ of smooth functions  on $(0,\infty)$ that are a partition of unity, in the sense that 
 \begin{equation}
 	\chi_j^2 + (\chi_j^\perp)^2 =1 \ . 
 \end{equation}
 The function $\chi_j$ 
 equals $1$ on $\mathcal{C}_j (E_j )$, and decays smoothly otherwise. 
 Similarly, $ \chi_j^\perp$  is   smooth  and 
 localized in the FRs. 
 In particular, the pair $(\chi_j , \chi_j^\perp)$ 
 is  constructed so  that its derivatives can be chosen to be  as small as desired. 
 Most importantly, 	
 the lower bounds   for the effective potentials  proven  in Subsection \ref{bcfr} can   be recast in terms
 of  $\chi^\perp_j$.
 This is the content of     Proposition~\ref{prop-chi-mu}.

\subsection{Localizing   the ARs}\label{bcfr}
In this subsection, we localize the classically allowed regions, asymptotically as $ | j | \rightarrow \infty$.
In order to do so, in the upcoming lemma, 
we prove lower bounds for the effective potentials
outside of an interval of the form $[\delta |j|^\sigma, c \< j\>^\sigma]$.

\vspace{1mm}
We remind the reader that 
$ \<x\> = (1+x^2)^{  \frac{1}{2}}$
stands for the standard  bracket,
and the inequality $M_1 \leq M_2$ between two-dimensional 
matrices is understood in the sense of quadratic forms on $\C^2. $

\begin{lemma}[Lower bounds on $V_j$]
\label{lemma lower bound V}
The following statements hold true. 
	\begin{enumerate}[label=(\roman*), leftmargin = *]
		\item For all  $\mu\in (0,1)$ there exists   $c_\mu  \geq 1$ such that for all  $c \geq c_\mu$ and 
		$j\in \Z$ 
		\begin{align}
			\label{right-bound}
			V_j (r ) \ge \mu A(r)^2  
			\quad 
			\mbox{for all }
			\quad 
			r \geq c  \<j\>^\sigma  \ . 
		\end{align}
		\item 
		 For all $\mu \in (0,1)$
		 and  $c\geq 1 $, 
		 there exists $\delta_{\mu,c}  \in (0,1)$
		 such that  
		for all  
		$\delta\in ( 0, \delta_{\mu,c} ]$ 
		and      $ |j| \geq 2$ 
		\begin{align}
						\label{lemma lower bound V 2}
			V_j (r ) 
			\geq 
			\mu A^2(2 c |j|^\sigma )
\quad 
\t{\textit{for all}}
\quad r \leq \delta |j|^\sigma\,.
		\end{align}
	\end{enumerate}
\end{lemma}

\begin{remark}
It will follow   from the proof, that  the parameters $c_\mu\geq 1$ and $\delta_{\mu,c} <1 $
	can be chosen to be of the form 
	(up to a constant that   depends only on $\alpha>0$)
	\begin{equation}
		c_\mu = 
 \Big(
 \frac{1}{1 - \mu}
 \Big)
  ^{   \max \big(  \alpha^{-1 }, 1+\alpha \big)	} 
		\qquad
		\t{and}
		\qquad
		\delta_{\mu,c}
		 =  
		 \min
\Big( 
 ( 1 - \mu)^\sigma \ , 
 \frac{1}{c^\alpha }
		 \Big) \ . 
	\end{equation} 
\end{remark}

\begin{proof}[Proof of Lemma \ref{lemma lower bound V}]
Note that for 
	$A(r) = r^\alpha$
	we obtain
	$B(r) = (1+\alpha) r^{\alpha -1 }$. 
	 We only give proofs for the component $V_j^+(r)$, the other one is analogous.
	
	\vspace{1mm}
	
\textit{(i)}
	Fix $\mu \in (0,1)$, and consider   $ c \geq 1 $,   
	 soon to be determined.
	Then, for $r \geq c \< j \>^\sigma$ 
	we    observe  
	that 
	$
	rA(r)
	\geq 
	c^ {1+\alpha}
	\<   j  \>   
	$
	and similarly
	$| ( B/A^2)  (r)| \leq (1+\alpha ) c^{-\alpha}$. 
	Consequently, 
	on the interval $[c\<j\>^\sigma , \infty )$
	we find 
	\begin{align}
		V_j^+  
		\   =  \ 
		A^2
		\Big[
		\Big(
		1 - \frac{1 + j}{r A}
		\Big)^2
		+ 
		\frac{B}{A^2}
		\Big]	
		\  \geq  \ 
		A^2 
		\Big[
		\Big(
		1 - \frac{ \sqrt{2}}{c^\sigma }
		\Big)^2
		- 
		\frac{1+\alpha}{c^\alpha}
		\Big]	 \  , 
	\end{align}	
	where we have used the fact
	$1 + |x| \leq \sqrt 2  \<x\>$. 
	A quick calculation shows that 
	the claim of the lemma
	  follows by choosing 
	$c 
	\geq 
	c_\mu  = 
	(2 \sqrt 2 + 1 + \alpha )^{ \max(1/\alpha, 1+\alpha)}
	(1 - \mu)^{-  \max(1/\alpha, 1+\alpha)}. 
	 $

	\vspace{1mm}
	
\textit{(ii)}
Fix $\mu\in(0,1)$	 and $c>0$.
Consider   $\delta \in (0,1)$ soon to be determined. 
Assume that
	$|j|\geq 2$ so that, in particular, 
	$|j|\leq 2 |j+1|. $
	Then, for $r \leq \delta |j|^\sigma$
	a tedious but straightforward calculation
	shows that  
	$|rA / (j+1)| \leq 2 \delta^{1+\alpha}$
	together 
	with
	$r^2B/ (j+1)^2 
	\leq
	4 (1+\alpha) \delta^{1+\alpha}
	$.
	Consequently, we can conclude that
	on the interval 
	$
	(0, \delta |j|^\sigma]
	$
	there holds
	\begin{align}
\nonumber 
		V_j^+ 
	& 	= 
		\bigg(  \frac{j+1}{r}		\bigg)^2
		\Big[
		\Big(
		1 -
		\frac{rA}{j+1}
		\Big)^2
		+ 
		\frac{r^2 B}{j+1}
		\Big] 	 \ , \\ 
		\nonumber 
	& 	\geq 
 \frac{|j|^2}{4 r^2}
		\Big[
		\Big(
		1 -
		2 \delta^{1+\alpha}
		\Big)^2
		-  
		4  (\alpha +1 )\delta^{1+\alpha}
		\Big]  \ ,  \\ 
		& \geq  
		\frac{ |j|^{2 \alpha \sigma }}{4 \delta^2}
		\Big[
		\Big(
		1 -
		2 \delta^{1+\alpha}
		\Big)^2
		-  
		4  (\alpha +1 )\delta^{1+\alpha}
		\Big]  \  , 
	\end{align}
where in the last line we used the inequality 
	$( |j|/r)^2 \geq \delta^{-2} |j|^{2 \alpha\sigma } $. 
First, we choose 
 $\delta \leq  [4(2+\alpha)]^{-\sigma} (1-\mu)^\sigma$
 so that the term in brackets $[ \cdots]$
 is larger than $\mu$. 
 Second, 
 we choose $ \delta \leq 1 / (2^{1 + \alpha} c^\alpha)$
 so that the pre-factor
 is bounded below by $A^2 (2c |j|^\sigma)$. This finishes the proof. 
\end{proof}

%

Lemma \ref{lemma lower bound V}
motivates the following definition.

\begin{definition}[The parameters $c_\mu$ and $\delta_{\mu,c}$    ]
	\label{definition delta mu c mu}
	Let $\mu \in (0,1)$ and $c \geq1$. 
	Then, we define   
	 $c_\mu\geq 1$
	 and 
	 $\delta_{\mu,c} $
	 according to the following criterion. 
	 \begin{enumerate}[leftmargin = * ]
	 	\item $c_\mu \geq 1$ is chosen so that the first statement of Lemma \ref{lemma lower bound V} holds true for 
	 	$\mu$. 
	 	
	 	\item $\delta_{\mu,c} <1 $ is chosen so that the second statement of Lemma \ref{lemma lower bound V} holds true for 
	 		$\mu$ and $c  $. 
	 \end{enumerate}
\end{definition}

 Let us finish this subsection with the   observation 
  that Lemma \ref{lemma lower bound V}
	implies that 
we can localize 
	the clasically allowed regions to an interval that scales like $ |j|^\sigma$.
Indeed, let $\mu \in (0,1)$
and consider  parameters
$c \geq c_\mu$
and
$\delta \leq \delta_{\mu,c}$. 
Then, for all 
  $|j|\geq 2$
	the  
	lower bound  
	$V_j(r)  \geq   \mu  |j|^{2\alpha\sigma}$ holds true 
	for $r \notin [\delta  |j|^\sigma ,  c  \<j\>^\sigma]$. 
	Therefore, 
	for all such $ j \in \Z$ we have that 
	\begin{equation}
		\label{localization CAR}
		E_j  < \mu |j|^{2 \alpha \sigma } 
		\implies 
		\calC_j (E_j)
		\subset 
		[	\delta  |j|^\sigma, c  \<j\>^\sigma 	] \ . 
	\end{equation}
	We read the above result as saying that the classically allowed regions 
	are \textit{stable} as one varies the energy levels.

%
%

\subsection{Partitions of  unity}\label{partitions}
In this subsection, we construct  two pairs of sequences of 
smooth localization functions. 
The first  one  we  denote by $   (\chi_j)_{j\in\Z}$, 
and can be though of as a sequence of localization functions around the
AR.
The second one is denoted by $  (\eta_j)_{j\in \Z}$, 
and should be thought of as a technical enlargement of $\bchi$.
In particular,
given  $\ve_0 >0 $,
the construction can be made    so that
their derivatives  are no larger than $\ve_0$. 

\vspace{1mm}
Let us note that the localization functions can be understood
as a smooth version of
the projection operators  (\textit{i.e.} sharp cut-offs)
first 
considered in the Schr\"odinger case  \cite{CHSV21}.
In the present case, smoothness is required since we study the square of the original Hamiltonian $H_0$.
This has the effect of incorporating commutator (\textit{i.e.} derivatives) into the analysis.

The construction starts as follows. 
First, given  two positive constants
$a< b $
we let  let 
$g_{a,b}  \in C_c^\infty (\R_+ , [0,1]) $
be a  smooth function of compact support that satisfies
\begin{align}\label{ll}
	g_{a,b}(r)=\left\{
	\begin{array}{cl}
		0\,,& r\in (0,a/2)\\
		\mbox{ monotone}\,, & r\in [a/2,a] \\
		1\,,& r\in(a,b)\\
		\mbox{ monotone}\,, & r\in [b,2b)\\
		0\,,& r\in [2b,\infty)
	\end{array}
	\right.\,.
\end{align}
In particular, 
we assume that 
it has been constructed in such a way so that:
\begin{enumerate}[label = (\roman*), leftmargin=.9cm ]
	\item its  complementary function 
	$g_{a,b}^\perp  \equiv \big(   1  - g_{a,b}^2  \big)^{1/2}$ 
	is smooth as well;  and 
	\item 
	$	
	\|	 g'_{a,b}	\|_{L^\infty } , \ 
	\|	 (g^\perp_{a,b}) ' \|_{L^\infty }  	  
	\leq  C_*/a $  
	for some  $C_*> 2$, 
	which we shall refer to as the \textit{universal constant}. 
\end{enumerate}
Observe that $\supp ( g_{a,b}^\perp) $ has two connected components: 
one of them is  supported on $[0,a]$, 
and the second one on  $[b,\infty)$.

\vspace{1mm}	
Secondly, for technical reasons, 
we need to  introduce another family of smooth   functions as follows. 
We define 
$h_{a,b}\in C_c^\infty (\R_+ , [0,1])$ 
to be such that 
\begin{align}
	h_{a,b}(r)=\left\{
	\begin{array}{cl}
		1 \,,& r\in (0, b )\\
		g_{a,b} (r) 		 \,, & r\in   [b,\infty)
	\end{array}
	\right.\,.
\end{align}
Let us  observe that the construction can be made so that: 
its associated complementary function 
$h_{a,b}^\perp $ is smooth as well; 
the analogous derivative bound holds true 
$\|	 h_{a,b} ' 	\| \leq C_* /b$; 
and $\supp (h_{a,b})$
has only one connected component. 

\vspace{1mm}

We   now     define the  sequences of smooth localization
functions as follows.

\begin{definition}
	\label{definition partition unity}
Let $\delta \in (0,1)$, $c\geq1 $
and 
	$j_ 0 \geq 2$.
	We refer
	to the following two
	sequences of functions 
	as the \textbf{partitions of unity}, 
	constructed with respect to $(\delta, c , j_0)$. 
	
	\vspace{1.5mm}
	\begin{enumerate}
		[leftmargin=* , label = (\roman*)]
		\item 
		For all $ j \in \Z $, we
		define  $ \eta =  (\eta_j)_{j\in \Z}$
		as the sequence of functions on $(0,\infty)$ given by 
		\begin{align}
			\nonumber
			\eta_j(r)
			\equiv 
			\begin{cases}
				h_{\delta ,c }   
				\Big( r / \<j\>^\sigma \Big)
				\,,& |j|\le j_0\\
				g_{\delta ,c }
				\Big( r / |j|^\sigma  \Big)
				\,,&|j|> j_0
			\end{cases} \ . 
		\end{align}
		We denote its complement   by
		$ 
		\eta^\perp_j 
		 \equiv 
		\big(
		1 - \eta^2_j
		\big)^{1/2 } .  $
		
		\vspace{1mm}
		
		\item
		For all $ j \in \Z $, we
		define  $ \chi = (\chi)_{j\in \Z}$
		as the sequence of functions on $(0,\infty)$ given by  
		\begin{align}
			\chi_j(r)
			\equiv 
			\begin{cases}
				h_{    \,   \delta /2  \,    ,2c  }
				\Big( r / \<j\>^\sigma \Big)
				\,,& |j|\le j_0\\
				g_{    \, \delta /2 \,     ,2c  } 
				\Big( r / |j|^\sigma  \Big)
				\,,&|j|> j_0
			\end{cases} \ . 
		\end{align}
		We denote its complement   by
		$ 
		\chi^\perp_j 
		 \equiv 
		\big(
		1 - \chi^2_j 
		\big)^{ 1/2} . $ 
	\end{enumerate}
\end{definition}

\begin{remark}
	In bold-face notation, 
	the direct sum of these localization functions
	are denoted by $\boldsymbol{\chi} = \oplus_{ j \in \Z } \chi_j $ 
	and
	 $\boldsymbol{\eta} = \oplus_{ j \in \Z } \eta_j $, respectively. 
\end{remark}

\begin{remark}
Fix $ j \in \Z$. Then, 	there holds $\eta_j \leq \chi_j$.
	In particular,  
	the  supports are nested into one another.
	Namely 
	$\supp (\eta_j) \subset \supp (\chi_j)$
	and similarly 
	$	 \supp (\chi^\perp_j) \subset \supp (\eta^\perp_j).$
\end{remark}
	%
	%
	%
	%
	%
	%
	%
	%
	%

Let us collect  in the following Proposition
the   important properties of the 
localization functions just introduced. 
Note that	the partitions of unity 
$\chi$
and $\eta$
strongly depend on the parameter 
$\delta$, $c$ and  $ j_0$, 
even though we do not display such dependence explicitly.

%

\begin{proposition} 
	\label{prop-chi-mu} 
Let 
$\delta \in (0,1)$, $c \geq 1$ and $j_0 \geq 2 $.
Consider  the partition of unity $(\chi_j)_{j\in\Z}$ and $(\eta_j)_{j\in \Z}$, 
constructed  with respect to $( \delta , c , j_0)$. 
  Then, the following statements hold true
  \begin{enumerate}[label = (\roman*)]
  	\item 
  	For all $ j \in \Z$  
  	\begin{align}
  				\label{bound error derivative}&
  		\|\eta'_j\|_{L^\infty},  \, 
  		\|\chi'_j\|_{L^\infty}, \, 
  		\|(\eta^\perp_j)' \|_{L^\infty}  \,  , 
  		\| 	 (\chi_j^\perp) 	\|_{L^\infty}
  		 \ \leq   \ 
  		  \, 
  		 C_* \max\bigg(
  		   \frac{1}{c} , \frac{1}{\delta j_0^\sigma }
  		 \bigg)
  	\end{align}

  	\item 
Given $\mu \in (0,1)$, let    $ (c_\mu, \delta_{\mu, c})$  be as in Definition \ref{definition delta mu c mu}.
  	Assume that 
 $c \geq c_\mu$
 and
 $\delta \leq \delta_{\mu,c}$. 
  	Then,  for all $ j \in \Z $
  	\begin{align}
  			\label{a-1}&
  			V_j(r)\,\ge \,\mu A^2(r) 
  			 &    &  \t{for all}\qquad 
  			 r\in {\rm supp}(\eta_j^\perp)\,,\\
  		\label{a-2}
  		&V_j(r)\,\ge \,\mu A^2(2 c  |j|^\sigma)
  		  		&   & 	 \t{for all}\qquad 
  		 r\in {\rm supp} (\chi_j^\perp)\,.
  	\end{align}

  \end{enumerate}

\end{proposition}

%

\begin{proof}[Proof of Proposition~\ref{prop-chi-mu}]
	Throughout this proof, we will fix $j\in \Z$.
	\vspace{1mm}
	
\textit{(i)}
 The estimates follow  from the chain rule,   the construction of
		$g_{a,b}$
		and 
		$h_{a,b}$, 
		and Definition \ref{definition partition unity}.

	\vspace{1mm}		

\textit{(ii)}
Let us prove \eqref{a-1}. 
	First, consider the case $|j|\leq j_0$. 
	Then it follows from 
	Definition
	 \ref{definition partition unity} 
	that 
	$\supp \, \eta_j^\perp = [  c  \<j\>^\sigma, \infty )$
	and  \eqref{a-1} now follows from 
	\eqref{right-bound} 
	in Lemma
	\ref{lemma lower bound V}. 
	Secondly, consider  the case $|j| >  j_0 $.
	Then, it follows from Def. \ref{definition partition unity}
	that 
	$
	\supp \eta_j^\perp 
	=
	[ 0, \delta |j|^\sigma ]
	\cup 
	[ c \<j\>^\sigma, \infty )
	$.
	It suffices to estimate the effective potential over the first component, 
	since the proof for the second is identical to the case $|j| \geq j_0$.
	Indeed,   it follows from Def. \ref{definition delta mu c mu}
	that 
	$ c \geq c_\mu $ 
	and
	$ \delta \leq \delta_{\mu,c}$  
 	imply 
	\begin{align}
		V_j(r) \,\geq\, \mu A^2(2c |j|^\sigma)\,\geq\,\mu A^2 (r)
		\qquad 
		\t{on}
		\qquad 
		 r\leq \delta  |j|^\sigma,
	\end{align}
	where in the last inequality we  have used the fact that $A(r)$ is monotone.

	\vspace{1mm}
	
Let us now prove \eqref{a-2}. 
	First, consider the case $|j| \leq j_0$.
	Then, $\supp \chi_j^\perp = [2 c \<j\>^\sigma, \infty )$
	and the claim follows from \eqref{right-bound} from Lemma 
	\ref{lemma lower bound V}. 
	Secondly, consider the case $| j| > j_0$
	so that in a similar way  
	$
	\supp \chi_j^\perp 	=	
	[ 0, (1/2)\delta  |j|^\sigma ]
	\cup 
	[ 2 c  \<j\>^\sigma, \infty )
	$. 
	The second component is estimated similarly as in the case $|j| \leq j_0$.
	The estimate on the second component
	follows from  
	$\delta \leq \delta_\mu  $
	and $c \geq c_\mu $,  Definition \ref{definition delta mu c mu} 
	and the fact that $A(r)$ is monotone. 
	%
\end{proof}

\begin{remark} 
It follows from   Proposition \ref{prop-chi-mu} 
that
the partition of unity has small derivatives, 
provided $c \geq 1 $ and $j_0$ are large enough. 
We understand  these derivative bounds as an 
error in localizing the kinetic energy of the system.
Let us further explain.
Namely, consider the following version of the 
\textit{IMS} localization formula, 
for $\vartheta \in\{\eta , \chi \}$, 
$\psi \in \Cr$   (see Definition \ref{definition domain cores})
and all $ j \in \Z $ 
\begin{equation}
	\label{localization formula}
	\< \psi ,  T  \psi  \> 
	=
	\< \vartheta_j \psi  , T \vartheta_j \psi  \>
	+
	\< \vartheta_j^\perp  \psi  , T \, \vartheta_j^\perp  \psi  \>
	- 
	\<  \psi  ,
	\big(  |\vartheta_j ' |^2 + |\vartheta_j^\perp\,   ' |^2  \big)
	\psi     \> \, .
\end{equation}
In particular, 
the last term  can be bounded in terms of the localization error, 
\textit{i.e.}
the following upper bound holds true 
$ |\vartheta_j ' |^2 + |\vartheta_j^\perp\,   ' |^2  \le 2\varepsilon^2$ , 
where $\ve = \ve( c , \delta, j_0)$ 
can be chosen 
as the right hand side of  \eqref{bound error derivative}. 
\end{remark}

As a corollary of Proposition \ref{prop-chi-mu}
and the localization formula \eqref{localization formula}, 
we obtain the following
estimate for 
the potential energy inside the FRs.

\begin{corollary} 
Let 
$\delta \in (0,1)$, $c \geq 1$ and $j_0 \geq 2 $.
Consider  the partition of unity $(\chi_j)_{j\in\Z}$ and $(\eta_j)_{j\in \Z}$, 
constructed  with respect to $( \delta , c , j_0)$. 
Let $\ve_0 > 0 $ satisfy 
\begin{equation}
\label{error inequality}
	\ve_0  \geq 
	   		 C_* \max\bigg(
	 \frac{1}{c} , \frac{1}{\delta j_0^\sigma }
	 \bigg)  \ . 
\end{equation}
Then, it holds in the sense of quadratic forms on 
	$\mathcal C    $ 
	that
	\begin{align}
		&
		\boldsymbol{\vartheta}^\perp 
		\boldsymbol{V}  
		\boldsymbol{\vartheta}^\perp 
		\leq 
		H_0^2
		+
		2\varepsilon_0^2 
		\label{v-for-h} 
	\end{align}
	where   
	$ \boldsymbol{V}   = \oplus_{ j \in \Z } V_j$, 
	and $\vartheta$ stands for either $\chi$ or $\eta$. 
\end{corollary}

\begin{proof}
	Let $j\in \Z$.
	Then, 
	the positivity  of  the operators  $T$
	and  
	$h_j^2$, 
	together with the   IMS formula 
	\eqref{localization formula}, 
imply that 
	in the sense of quadratic forms on $\mathscr C _r$
	\begin{align}
		\nonumber
		\vartheta_j^\perp 	V_j  	 \vartheta_j^\perp
 	\leq 
		\vartheta_j^\perp (V_j + T )\vartheta_j^\perp 		 
		\nonumber
	  = 
		\vartheta_j^\perp  h_j^2 \vartheta_j^\perp 	 
		\nonumber
		  \leq 
		\vartheta_j^\perp  h_j^2 \vartheta_j^\perp 
		+
		\vartheta_j  h_j^2 \vartheta_j 	 
		\nonumber
		 = 
		h_j^2 
		+ 2 
		\big(  |\vartheta_j ' |^2 + |\vartheta_j^\perp\,   ' |^2  \big)	 \ . 
	\end{align} 
	The proof is finished once we use \eqref{bound error derivative}, 
	\eqref{error inequality}, 
	and   sum over $ j \in \Z $. 
\end{proof}

	\section{Relative bounds for the perturbation}
	\label{section perturbation bounds}

Under the assumptions of  Condition~\ref{condition 2}
the operator  $W$ is possibly unbounded at infinity.  
Hence, 
one may naturally ask if  the operator $H_0$ is able to dominate such perturbations. 
It turns out that -- in the general case -- this is not possible. 
Heuristically, this follows from the fact that
the magnetic Hamiltonian $H_0$
vanishes in phase-space points where $\bp = \bA(\bx)$. 
From this equation one may extract a path to infinity in which $W(\bx)$ diverges.

In this section 
we consider 
 a novel approach that allows to control  the  perturbation $W$ in terms of $H_0$
	plus an additional term, localized over  the classically allowed regions, that is diagonal in the angular momentum channels.
	This is the content of Theorem \ref{theorem W}.

	\subsection{Bounds for the potential}
	\label{subsection bounds potential}
Throughout this section, we consider $A(r)$ as in Condition 
\ref{condition 1'}. 
Further, we  consider  a perturbation $W$ 
that 
verifies Condition 
\ref{condition 2}
with respect to the  parameters $\mu_*\in(0,1)$ and $\beta>0$, 
and  with potential $v\in L^2_{\rm loc} ([  0 , \infty), \R  	)$.

	In order to gain control over $W$ we introduce the following sequence  of  positive numbers.  
	Heuristically, these numbers $\lambda_j$  control  the values that $W$ takes 
	over the phase space points $\bp = A (\bx )$ when restricted to the angular momentum channel $j\in\Z$.

	\begin{definition} 
		\label{definition lambda}
Let $ c \geq 1$. 
Then, 		for all $ j \in \Z$ we define the following sequence of 
positive numbers 
	\begin{align}\label{lambda}
			\lambda_j 
			\equiv
			 A 
			 (2c 		\<j\>^\sigma )
			\,.
		\end{align}
We denote  its direct sum 
by
$
 \boldsymbol{\lambda}
 \equiv \oplus_{ j \in \Z } \lambda_j$. 
	\end{definition}

	The main result of this section is the inequality contained in the following theorem. 
	The first step is
	to reduce the study from  $W$
	to the  potential $v_\beta =\coth(\beta/4) v$; this is achieved using Lemma~\ref{lemma exp twisted}. 
	Next we consider the decomposition into a 
	\textit{critical}
	and
	\textit{non-critical}
	part, regarding the growth at infinity. Namely
	\begin{equation}
		\label{v beta decomposition}
		v_\beta^2 = \mu_* A^2 + \omega^2\,. 
	\end{equation}
	In view of \eqref{eq.lim} we have that $\omega (r) = o (A(r))$ as $r  \rightarrow \infty$.

	\begin{theorem}
		\label{theorem W}
Let $\mu, \delta, \ve_0 \in (0,1)$ and $j_0, c  >  1$.
Further, assume that:  $\mu \in (\mu_*, 1 )$; 
$c \geq c_\mu$ and $\delta \leq \delta_{\mu,c}$ 
(see Def. \ref{definition delta mu c mu}); 
and 
$\ve_0 >0$
satisfies the lower bound  \eqref{error inequality}
in terms of $(c, \delta, j_0)$. 
		Let $\chi$ be a partition of unity, constructed with respect to 
		with respect to $(c, \delta, j_0)$ (see Definition  \ref{definition partition unity}). 
		Then, for all $\ve_*, \, \ve_1 \in (0,1)$
		there exists $C_0>0$
		such that
		\begin{equation}
\label{thm W eq 1}
			W^2 
			\, 	 \leq 	\, 
			v_\beta^2  
			\, 	\leq  \, 
		\bigg( 
			\frac{\mu_* + \ve_1 + \ve_*  }{\mu}
			\bigg)  
			H_0^2 
			+ 
(\mu_*+ \ve_* )
			\boldsymbol{\chi}
			\boldsymbol{\lambda}^2  
			\boldsymbol{\chi}
			+ C_0  
		\end{equation}
		in the sense of quadratic forms on  $\mathcal C $.
Moreover, the bound \eqref{thm W eq 1} holds true with $\ve_*=\ve_1=0 $ for  bounded $\omega$, and with
 $\ve_*=0$ whenever $\limsup_{r \to \infty} |\omega(r)|=0$.
				\end{theorem}
		
		\begin{remark} 
\begin{remarklist}
	\item
The contributions containing $\mu_*$
arise from the critical part $\mu_* A^2$
in \eqref{v beta decomposition}. 
On the other hand, 
the infinitesimal contributions 
containing $\ve_*$
and
$\ve_1$
arise from the non-critical part $\omega$.
In particular, 
recall that 
			$\mu_*  < \mu  < 1$.
			Consequently, 
			we can choose $\ve_1$ and $\ve_*$
			to be small enough so that 
both pre-factors in Theorem \ref{theorem W}
are smaller than one. 
			This fact will be crucial in our analysis in the next section. 
\item In applications we will use \eqref{thm W eq 1} for 
$\ve_\ast,\ve_1\in[0,1)$ with the convention that $\ve_1, \ve_\ast$ take the value zero when $\omega$ satisfies the conditions    prescribed above.
\end{remarklist}
		\end{remark}

	\begin{proof}[Proof of Theorem \ref{theorem W}]
		The first inequality contained in Eq. \eqref{thm W eq 1}
		follows from Lemma
		 \ref{lemma exp twisted}
		for $ F \equiv 0$.
		As for the second inequality, 
		we start from the decomposition contained in \eqref{v beta decomposition}
		and use Lemma \ref{lemma critical part} to control the critical part, 
		and Lemma \ref{lemma non critical} to control the non-critical part. 
It then suffices to put everything together. 
	\end{proof}
	We dedicate the remainder of this section
	to the proof of Lemma \ref{lemma critical part}
	and \ref{lemma non critical}.

\subsubsection{The critical part}
	Our first result is the following Lemma, 
	which controls the critical part $\mu_* A^2$. 
	It shows 
	that   the   magnetic vector potential 
	$A^2$ 
	can be dominated by two different contributions: 
	one given by $H_0$ and     one given by 
	$\blambda $

	\begin{lemma}
		\label{lemma critical part}
Under the same assumptions of Theorem \ref{theorem W}, 
the  following inequality holds   
	in the sense of quadratic forms of 
$\mathcal C$ 
		\begin{align}
			A^2
			\leq 
			 \frac{1}{\mu} 
H_0^2 
			 +
\boldsymbol{\eta}
\boldsymbol{\lambda}^2
\boldsymbol{\eta}
			  +
			\frac{2  \varepsilon_0}{\mu}  
				\,. \label{a-for-h}
		\end{align} 
	\end{lemma}

	\begin{proof}
Let $\eta$ be the partition of unity
in the sense of Definition \ref{definition partition unity}.
We consider the decomposition
	\begin{equation}
\label{decomposition of A}
		A^2 = 
			\boldsymbol{\eta}^\perp
		A^2
		\boldsymbol{\eta}^\perp
		+
		\boldsymbol{\eta}
		A^2
				\boldsymbol{\eta}
	\end{equation}
and study each term separately. 
The first term in \eqref{decomposition of A}
can be estimated using Proposition \ref{prop-chi-mu}, 
combined with the localization estimate \eqref{v-for-h} for the effective potentials
$V = \oplus_{ j \in \Z }V_j$. 
We obtain 
that in the sense of quadratic forms in $\mathcal C$
\begin{equation}
				\boldsymbol{\eta}^\perp
	A^2
	\boldsymbol{\eta}^\perp
	\leq 
	\frac{1}{\mu}
					\boldsymbol{\eta}^\perp
V
	\boldsymbol{\eta}^\perp
	\leq 
		\frac{1}{\mu}
	H_0^2 + 
		\frac{2\ve_0 }{\mu}  
		 \ . 
\end{equation}
The second term in \eqref{decomposition of A}
can be estimated in terms of 
the sequence $ (\lambda_j)_{j\in\Z} $, 
and the fact that
$\sup \supp \eta_j = 2 c  \<j\>^\sigma$. Namely, we obtain
\begin{equation}
		\boldsymbol{\eta}
A^2
\boldsymbol{\eta}
\leq 
		\boldsymbol{\eta}
\boldsymbol{\lambda}^2
\boldsymbol{\eta} \ . 
\end{equation}
Combining the two contributions finishes the proof. 
	\end{proof}

  \subsubsection{{The non-critical part}}
Let us now control the non-critical part 
$\omega$ of the potential.
\begin{lemma}
	\label{lemma non critical}
Under the same assumptions of Theorem \ref{theorem W}, 
	let  $\omega$ be given through \eqref{v beta decomposition}. 
 	Then, for all $\ve_1, \ve_* \in (0,1)$
 there exists a constant $C>0$ such that 
 in the sense of quadratic forms on 
 $\mathcal C$
 \begin{equation}
 	\omega^2 
 \  	\leq  \ 
 \bigg( 	\frac{\ve_1 + \ve_*}{\mu}		\bigg) 
 	H_0^2
 \,  	+ \, 
 	\ve_* 
 	\boldsymbol{\chi}
 	\boldsymbol{\lambda}^2
 	\boldsymbol{\chi}
\,  	+  \, 
 	C  \ . 
 \end{equation} 
Moreover, if $\limsup_{r\to\infty} |\omega(r)|=0$, the above inequality holds for $\varepsilon_\ast=0$.
\end{lemma}

\begin{proof} 	
First, we control $\omega$ at infinity.
Since $\omega(r)=o(A(r))$, for any $\varepsilon_\ast  \in [0,1)$ there is   $R_0>0$ and  $C_{\varepsilon_\ast}>0$ such that
		\begin{align}
			\label{mono1}
\1_{	[ R_0  ,\infty )	}
			\omega^2\le \varepsilon_\ast A^2 +C_{\varepsilon_\ast}\,. 
		\end{align}	
On the other hand, 
since $A\in L^2_{\rm loc}(\R^2)$  we have, by \eqref{v beta decomposition},  that $\omega$ belongs to the same space.
Therefore, by Lemma~\ref{loc.int}
we have  that, for all  $\varepsilon_{\rm l} \in (0,1)$, there is a constant $C_{\varepsilon_{\rm l}}>0$ such that 
		\begin{align}
			\label{mono2}
			\1_{[0,R_0]}	
			\omega^2
			\le
			\frac{\varepsilon_{\rm l}}{\mu } H_0^2+C_{\varepsilon_{\rm l}}
			\,.
		\end{align} 
Finally, we consider the decomposition
$
	\omega^2 
 = 
\1_{	[ R_0  ,\infty )	}
\omega^2
+ 
\1_{[0,R_0]}	
\omega^2 
$ and 
use \eqref{mono1}, \eqref{mono2}
as well as   Lemma \ref{lemma critical part}
to conclude that 
\begin{align}
\nonumber
	\omega^2  
 \ \leq  \ 
\ve_* A^2  
+ 
\frac{\ve_1}{\mu} H_0^2 + C_{\ve_1 } +C_{\varepsilon_\ast}	 
 \  \leq \ 
\frac{(\ve_* + \ve_1)}{\mu }
H_0^2
+ 
\ve_* 
\boldsymbol{\chi} 
\boldsymbol{\lambda}^2
\boldsymbol{\chi} 
+
C  
\end{align}
for a constant $C>0$. The last statement follows since if $\limsup_{r\to\infty} |\omega(r)|=0$ we may take $\varepsilon_\ast=0$ in \eqref{mono1}.
This finishes the proof. 
\end{proof}

\subsection{}

%
%

%
%
%

	 \section{Exponential decay and the proof of Theorem~\ref{theorem tunneling estimates}}
	 \label{section abs decay thm}
 In this section we state an abstract result
regarding   the  exponential decay of the eigenprojection  $\p_I (H  )$, 
 away from the   the support of $\chi^\perp=(\chi_j^\perp)_{j\in \Z}$
 or, in other words, 
away from  the classically forbidden regions. 
 This is the content of   Theorem~\ref{theorem exp decay}, 
 and its proof is postponed to the next section. 
Here, we focus on how  our main result, Theorem~\ref{theorem tunneling estimates}, 
follows from this abstract result. 
 Without loss of generality we consider $A(r)$ as in Condition \ref{condition 1'}.

	 \subsection{Exponential weights}
	 In order to quantify   exponential decay  we  introduce an appropriate class of operators playing the role of \textit{exponential weights}.
	 To this end,  we consider the following space of sequences of real-valued, radial functions
	 \begin{equation}  
	 	\W  
	 	: =
	 	\big\{  \, 
	 	F = (F_j)_{ j \in \Z} 
	 	\  \big|   \ 
	 	F_j \in W^{1,2}_{ \t{loc}}  \big(  ( 0,\infty ) ,  \R    \big) \ \forall j \in \Z 
	 	\, 
	 	\big\} \ . 
	 \end{equation}
	 To any $F =  (F_j)_{j\in \Z} \in \W $, we associate the following two operators on $ \H $  
	 \begin{equation}  
	 	e^{\pm \boldsymbol{F} } 
	 	:=
	 	\sum_{j\in\Z} 
	 	e^{\pm F_j}(|\bx |) P_{m_j}
	 \end{equation}
	 which we define  on their respective maximal dense domains 
	 $\calD(e^{\pm 
	 	\boldsymbol{F}
 	}
 	)$. 
	 Both operators are self-adjoint, positive and inverses of each other. 
	 In addition, we define the space of \textit{bounded weights} on $\mathcal{H}$ as 
 \begin{equation}
	 \W_b:=\{F \in \W \ \big| \ \|F\| 
= 
\sup_{j\in\Z} 
\sup_{r>0}
|F_j(r)|  
< \infty \}\,. 
 \end{equation}
Note that for bounded weights, the associated exponential operators
$e^{\pm 	\boldsymbol{F}}$ are   bounded operators on $\H$.

\subsection{Energy levels}
One of the main ingredients in our analysis 
are the Clasically Allowed Regions,  introduced in Definition \ref{definition CAR}
and further analyzed in Section \ref{section analysis car}. 
These
 are defined in terms of  a sequence of positive real numbers $ (E_j)_{ j \in \Z}$, that we regard as  reference energy levels. 
Let us now specify how this sequence looks like.

\vspace{1mm}
First, 
 we specify 
a collection of   parameters that have been introduced in the previous section.
Namely, we consider  
nine   non-negative real numbers
$(\mu , \delta , c,  j_0 ,   \ve_0 , \ve_1 ,\ve_*, C,U   )\in \R_+^9 $, 
which we shall refer to as the \textit{internal parameters}. 
These are going to be tuned 
with respect to four  
\textit{external parameters} of our theory: 
$(\alpha, \beta , \mu_*,E )$. 
The first three are determined by the Dirac Hamiltoninan $H_0$ satisfying Condition 
\ref{condition 1} and \ref{condition 2}, 
whereas $E  \geq  1$ is understood as the energy of a state $\vp\in\H$ under consideration. 

\begin{condition}[Parameter tuning]
	\label{condition tuning}
	Let 
	$
	(\mu , \delta , c,  j_0 , \ve_0 , \ve_1, \ve_*   ,C,U,  )\in [0,\infty)^9  
	$
be a collection of internal parameters, 
and let us introduce the following notations
\begin{align}
\begin{cases}
	 & 	
	  \tilde \mu  \equiv \mu_* + \ve_1 + \ve_*   
	\qquad
	\t{and}
	\qquad 
	\bar \mu_*
	 \equiv
	   \mu_*+ \ve_*   \ ,   \\ 
&
	a \equiv  
	 (		\widetilde \mu / \mu )^{\frac{1}{2}}  , 
\qquad
b 
\equiv 
   ( 	\overline{\mu}_*	 /   \widetilde \mu )^\frac{1}{2} , 
\qquad 
\t{and}
\qquad
\delto  
=  
\min(
1 - a, 1 - b
)   \ ,  \\ 
& \widetilde C_0 \equiv  (1-a) \ve_0 + C /a  \ , \\ 
&
\label{definition widetilde E}
\widetilde E ^2 
\equiv 
\big(
 1			-			\ve_0/ \delta_0 
  \big)
U^2  	-		 2 \ve_0 \delta_0^2 		- 
\delta_0  C/a     \ . 
\end{cases}
\end{align}
We say that the internal parameters 
	are {\textbf{tuned}} with respect to 
	$(\alpha , \beta  , \mu_*, E) \in \R_+^4 $
	provided the following conditions are satisfied. 
\end{condition}

\begin{enumerate}[label =  (\roman*) , leftmargin  = 0.7 cm  ]
\item   \textit{Basic inequalities.}
\begin{align}
\label{po}
\tag{P0}
	 \mu \in (\mu_* ,1 ) ,  \   \ 
	 \delta \in (0,1) , \  \ 
	 c \geq 1  , \  \ 
	 j_0 \geq 1  , \  \ 
	 U \geq 1   , \  \ 
	 \ve_0 \in  \Big(0,\frac{1}{16 }	\Big) 
	 \  \t{ and } \ 
	 \ve_1 , \   \ve_* \in [0,1)  . 
\end{align}

\item  \textit{AR constraints.}
Letting 
 $( c_\mu , \delta_{\mu, c})$
be  the parameters from  Definition \ref{definition delta mu c mu} 
\begin{align} 
\label{p1}
\tag{P1}
c \geq c_\mu  
\quad \t{ and } \quad 
\delta \leq 
\delta_{\mu,c} \ .
\end{align}

\item  \textit{Localization error.}
In the context of Proposition \ref{prop-chi-mu}
\begin{align} 
\label{p2}
\tag{P2}
  \ve_0  \geq
   C_* 
   \max\Big(	 \frac{1}{c} , \frac{1}{\delta j_0^\sigma }	\Big) \ . 
\end{align} 

\item \textit{$W$ constraints.}
In the context of Theorem \ref{theorem W} 
\begin{align}
\label{p3}
\tag{P3}
 	 \tilde \mu     <  \mu, 
 	 \qquad 
 \bar \mu_*    < 1 \   , 
 \quad 
 \t{and}
 \quad 
 C \geq C_0 \ . 
\end{align}


\item 
	\textit{Energy constraints on $U$ and $\ve_0$. }
\begin{align}
\tag{P4}
\label{p4}
\begin{cases} 
& 	\widetilde E  
  \ 	\geq  \ 
	E + \ve_0   \ ,   \\
	& 	(1   -  2 \ve_0/ \delto )  U^2 
	 \ \geq  \ 
	2 \delta_0 \widetilde C_0  \ ,   \\ 
& 
U \ \geq   \ 
4  \, 
\sup \{ |z|: z   \in \C , \ | \t{Re}\,  z|\leq E + \ve_0 , \ |\t{Im} \, z|\leq \ve_0  \}   \ . 
\end{cases} 
\end{align}

    \item \textit{Energy constraint on $c$. } 
\begin{align} 
	\label{p5}
	\tag{P5}
		 \mu ( 1  - \sqrt{	\bar \mu_* / \mu 	} ) A^2 ( 2 c )	
		 \  \geq \   2   U^2 / \delta_0^2  \ . 
\end{align}

	\end{enumerate} 

{  
	 \begin{remark}[Existence]
Note  that for all  $(\alpha, \beta, \mu_*, E )$
one can always find parameters 
$	(\mu , \delta , c,  j_0 , \ve_0 , \ve_1, \ve_*   ,C,U  ) $
satisfying Condition \ref{condition tuning}. 
To see this, fix first any  $\mu \in (\mu_*,1 )$.
Then, fix   $\ve_1$ and $\ve_*$
such that  \eqref{po} and \eqref{p3} hold.
Next, fix  $C  =  C_0$ in \eqref{p3}
together with  $\ve_0 < \frac{\delta_0}{16} .$
Thanks to the last bound, we can fix  $U = U (E, C, a)$
such that \eqref{p4} holds. 
Finally, fix $\delta$ small enough, and $c$ and $j$ large enough
such that 
  \eqref{po}, \eqref{p1},  \eqref{p5} and \eqref{p2} hold. 
 \end{remark}
}

\vspace{2mm}

We are now ready to define the energy levels that we work with.  
Let us recall here that
 the sequence $\lambda = (\lambda_j)_{ j \in \Z}$
 has been introduced in
 Def. \ref{definition lambda} .

\begin{definition} 
	\label{definition-energies}
The  
	 \textbf{reference energy levels}
	 corresponds to the sequence of positive numbers
	$  (E_j)_{j\in \Z}$,  
	defined through the relation
	\begin{equation}\label{Ej definition}
		E_j^2 
		\ 	\equiv \ 
		( 	\mu \mast)^{\frac{1}{2}}
		\, 
		\lambda_j^2 
		\, + \,  	\frac{1}{\delta_0^2 } U^2 \  , 
		\qquad \forall j \in \Z  
	\end{equation}

\end{definition}

\begin{remark}
A few comments are in order.
\begin{enumerate}[leftmargin=0.7cm]
	\item 
It follows from the definition of $\widetilde E$ in \eqref{definition widetilde E}, 
the first condition in  \eqref{p5},  and  the fact that $\delta_0 < 1 $, 
that for all $ j \in \Z $  the energy levels are bounded below by the reference energy
\begin{equation}
	E_j > E  \ . 
\end{equation}
In applications, it will be useful to consider   $ E \geq 1 $ in order to simplify 
algebraic manipulations.

	\item 
		If $W$ is a bounded operator, 
	it is possible to choose 
	$\ve_1 = \mu_\ast = \mast = 0$.
	In this case, 
	the reference  energy levels 
	can be chosen 
	uniformly in   $j \in \Z $.
	This is the analogous situation 
	that was originally considered in the non-relativistic case
	\cite{CHSV21}.

	\item 
	We  say that
	the reference  energy levels are \textit{congruent}
	with respect to the partition of unity   $\chi$. 
This comes from the observation that the operator inequalities hold 
	\begin{equation}
		\label{eq: Ej admissible}
		\bchi^\perp		
		( \boldsymbol{V}  - \boldsymbol{E}^2) 
		 \bchi^\perp 
		\ 	 \geq  \ 
		0 \    \ .
	\end{equation}
	In particular, \eqref{eq: Ej admissible} implies that
	$\mathcal{C}_j(E_j) \subset \supp \chi_j$ for all $ j \in \Z$.
	In order to prove the above inequality, 
	we apply Proposition \ref{prop-chi-mu}
	to find that 
	\begin{align}
		\bchi^\perp
		\big(
		 \boldsymbol{V}  - \boldsymbol{E}^2 
		\big)
		\bchi^\perp
		\geq 
		\bchi^\perp 
		\big(
		\mu  \blambda^2
		- 
		\boldsymbol{E}^2 
		\big) 
		\bchi^\perp  \ . 
	\end{align}
	In our parameter regime, 
	the sequence defining   the  above right hand side is strictly positive.
This follows from the inequalities 
	\begin{align}
		{   
			\mu \lambda_j^2 
			- E_j^2
			\  =  \ 
			\mu 
			\Big(
			1 - 
(\bar \mu_* / \mu)^{\frac{1}{2}} 
				 \Big)
			\lambda_j^2 
			- 
			 \frac{1}{\delta_0^2 }
			  U^2 
			\ 	 \geq  \ 
			\mu 
		\Big(
		1 - 
		(\bar \mu_* / \mu)^{\frac{1}{2}} 
		\Big)
		\lambda_0^2 
		- 
		\frac{1}{\delta_0^2 }
		U^2 
			\geq 0  \ .
		} 
	\end{align}
	Here, we  have used the definition of $E_j$, 
	the lower bound  
	$\lambda_j 
	\geq \lambda_0 =  A^2 (2 c ) $, 
and  the inequality contained in 
	\eqref{p5}. 
\end{enumerate}

\end{remark}

	 \subsection{An abstract result on exponential decay}
	
	Our proof of the  tunneling estimates contained in Theorem \ref{theorem tunneling estimates}
	    relies on   the following   abstract exponential decay result.

	 \vspace{1mm}

	 \begin{theorem}\label{theorem exp decay}

Let us consider the following. 

\begin{itemize}[label= - , leftmargin =  * ]
	\item 
		{  		Let $H $ be the Dirac Hamiltonian   satisfying 
		Condition \ref{condition 1'} with $\alpha>0$
		and Condition \ref{condition 2} with    $\beta$, $\mu_*$.
		Further,  let  $E>0$. } 
	
	\item 
	Let 
	$(\mu , \delta , c,  j_0 ,   \ve_0 , \ve_1 , \ve_* , U, C ) $
	be the internal parameters, tuned with respect to 
	$(\alpha , \beta , \mu_* , E)$ in the sense of Condition \ref{condition tuning}. 
	
	\item 
Let
	$(\chi,\eta)$
	be a partition of unity, built with respect to $(c, \delta, j_0)$
in the sense of Definition \ref{definition partition unity}. 
	Let $(E_j)_{j\in \Z}$ be the energy levels as in 
Definition \ref{definition-energies}. 
\end{itemize}
Then, for any exponential weight $F \in\W  $ 
  that satisfies 
	 	\begin{align}
\tag{F1}
\label{F1}
	 		&    \boldsymbol{F}   \bchi  =  0     \\
\tag{F2}
\label{F2}
	 		& (  \boldsymbol{F}   '  )^2  
	 		\leq
\bigg(
\frac{1 -a }{2}
\bigg)^2
\bchi^\perp
(\boldsymbol{V} - \boldsymbol{E}^2 )
 \bchi^\perp  \\ 
\tag{F3}
\label{F3}
	 		& 
	 		\sup_{r>0 }
	 		|F_j - F_k|
	 		\leq
	 		(\beta /2 ) 
	 		| j - k |    \ , \quad \forall j , k \in \Z \  , 
	 	\end{align}
	 	 	 	it holds true that  $e^{ \boldsymbol{F} }  \p_{ [-E,E ] } (H  ) $ is a bounded operator on $\H$.
	 \end{theorem}

	 \subsection{Proof of Theorem~\ref{theorem tunneling estimates}}
In order to prove Theorem \ref{theorem tunneling estimates}
we shall choose explicit sequences of weights $F$ and $G$
that satisfy the conditions  \eqref{F1}, \eqref{F2} and \eqref{F3}. 
It turns out that \eqref{F2} is the most difficult condition to check in practice. 
The next lemma gives      lower bounds 
on ${\boldsymbol \chi}^\perp ({\boldsymbol V }- {\boldsymbol E}^2) {\boldsymbol \chi}^\perp $
that both motivates the definition of the upcoming weights, 
and helps checking that the condition is verified.

	 \begin{lemma}\label{lemma: Ej admissible}
	 	Under the same assumptions of Theorem \ref{theorem exp decay}, 
  the following statements hold true. 
\begin{enumerate}[label = (\roman*)]
	\item For $|j| > j_0$ and $r \leq \delta |j|^\sigma /4 $ there holds 
	\begin{align}
\label{lemma Vj 1 eq 1}
		\chi_j^\perp (V_j  - E_j^2) \chi_j^\perp 
		 \geq  
		 \frac{ \mu \delta_0}{2} 
		 \lambda_j^2  \ . 
	\end{align}

	\item For  $j\in\Z$ and $r \geq  4 c  \< j \>^\sigma $ there holds 
		\begin{align}
			\label{lemma Vj 1 eq 2}
		\chi_j^\perp (V_j  - E_j^2) \chi_j^\perp 
		\geq  
		\frac{ \mu \delta_0}{2} 
A^2(r)  \ . 
	\end{align}
\end{enumerate} 
	 \end{lemma}

 \begin{proof}
 		We start with a general inequality, and then show how $(1)$ and $(2)$
 	follow from it. 
 	Namely, we start by decomposing in operator notation 
 	\begin{align}
 		\nonumber 
 		\bchi^\perp  		( \bV - \bE^2)		\bchi^\perp 
 		&  = 
 		\bchi^\perp 
 		\Big( \bV -    			 	( 	\mu	\mast)^{\frac{1}{2}} \, \blambda^2 		\Big)
 		\bchi^\perp  
 		-  \delta_0^{-2}
 		\bchi^\perp  U^2 \bchi^\perp  \\
 		& = 
 		\frac{1}{2}
 		\bchi^\perp 
 		\Big(\bV -  
 		( 	\mu \mast)^{\frac{1}{2}}
 		\, \blambda^2 	
 		\Big)
 		\bchi^\perp  
 		+
 		\bchi^\perp 
 		\Bigg(
 		\frac{1}{2}
 		\Big(
 		\bV -   			 	( 	\mu \mast)^{\frac{1}{2}} \, \blambda^2 	
 		\Big)
 		-
 		 \delta_0^{-2}
 		U^2 
 		\Bigg)
 		\bchi^\perp 
 		\label{lemma V lower bound eq 1} \ . 
 	\end{align}
 	Let us now show that the second term of
 	\eqref{lemma V lower bound eq 1}
 	is non-negative.
 	Indeed, Proposition \ref{prop-chi-mu}, 
 	$\lambda_j \geq \lambda_0 = A(2 c )$
 	and 
\eqref{p5}
 	imply that 
 	\begin{align}
 		\bchi^\perp 
 		\Bigg(
 		\frac{1}{2}
 		\big(\bV - 
 		( 	\mu \mast)^{\frac{1}{2}}
 		\, \blambda^2 		\big)
 		-
 	\frac{U^2}{
 	\delta_0^2}
 		\Bigg)
 		\bchi^\perp 
 		& \geq 
 		\bchi^\perp 
 		\Bigg( 
 		\frac{\mu }{2}
 		\big(
 		1 - (\bar \mu_* / \mu )^{\frac{1}{2}}
 		\big) 
 		\blambda^2 
 		-
 	\frac{U^2}{
	\delta_0^2}
 		\Bigg)   
 		\bchi^\perp 
 		\geq 0 \ . 
 		\label{lemma V lower bound eq 2}
 	\end{align}
 	We are now ready to specialize into the two different items of the proof.
 	
 	\vspace{1mm}
 	
\textit{(i) }
 	In view of \eqref{lemma V lower bound eq 2}
 	it suffices to look at the first term of 
 	\eqref{lemma V lower bound eq 1}.
 	To this end, we apply 
 	Proposition \ref{prop-chi-mu}
 	to find that for all $ j \in \Z $
on 
 	$  r \geq 4 \<j\>^\sigma $
 	there holds 
 	\begin{align}
 		\chi^\perp_j 
 		V_j
 		\chi^\perp_j 
 	 \ 	\geq  \ 
 		\mu 
 		\chi^\perp_j 
 		A^2
 		\chi^\perp_j 
\  		= \ 
 \mu  A^2 
 	\end{align}
 	where we have used that $\chi_j^\perp(r) =1 $
 	for $r > 2 c  \< j\>^\sigma$, see Def. \ref{definition partition unity}. 
 	On the other hand, 
 	for $r \geq 2 c  \<  j \>^\sigma $
 	there holds $A^2(r) \geq A^2(2 c   \<  j \>^\sigma) = \lambda_j^2$.
 	Thus, we conclude 
 	that on 
   $r  \geq 2 c \<  j  \>^\sigma $ 
 	there holds 
 	\begin{equation}
 		\chi_j^\perp 	\lambda^2_j \chi_j^\perp 
 	 \ 	\leq \ 
 		\chi_j^\perp 	 A^2   \chi_j^\perp  
 \  		=  \ 
 		A^2
 	\end{equation}
 	where  again we used that $\chi_j^\perp(r) =1 $.
 	We can conclude that 
on    $r  \geq 2 c \<  j  \>^\sigma $   there holds 
 	\begin{align}
 		\chi^\perp_j 
 		\Big(V_j -  
 		( 	\mu \mast)^{\frac{1}{2}}
 		\, \lambda^2_j
 		\Big)
 		\chi^\perp_j  
 		\geq 
 		\mu A^2 - 
 		( 	\mu \mast)^{\frac{1}{2}} 
 		A^2 
 		=
 		\mu 
 		\Big(
 		1  -  \sqrt{\frac{\bar\mu_*}{\mu }} 
 		\Big)
 		A^2 \ . 
 	\end{align}
 	It suffices now to go back to
 	\eqref{lemma V lower bound eq 1}.
 	
 	\vspace{1mm}
 	
 \textit{(ii)}
 	Proposition \ref{prop-chi-mu}
 	immediately implies that
 	for all $ j \in \Z$ 
 	\begin{align}
 		\frac{1}{2}
 		\chi^\perp_j 
 		\Big(V_j -  
 		( 	\mu \mast)^{\frac{1}{2}}
 		\, \lambda^2_j
 		\Big) 
 		\geq 
 		\frac{\mu }{2 }
 		\Big(
 		1  -  \sqrt{\frac{\bar\mu_*}{\mu }} 
 		\Big) 
 		\chi^\perp_j
 		\lambda^2_j
 		\chi^\perp_j 
 		\label{lemma V lower bound eq 3}.
 	\end{align}
 	Now, for $| j| > j_0$
 	and $r   \leq \delta /2 |j|^\sigma$
 	there holds  
 	$\chi_j^\perp(r)=1$.
 	This observation
 	combined with 
 	\eqref{lemma V lower bound eq 1}, 
 	\eqref{lemma V lower bound eq 2}
 	and
 	\eqref{lemma V lower bound eq 3}
 	finishes the proof. 
 \end{proof}

%

{ 
	 	 We will consider the following family of exponential weights, that are inspired  by \cite{CHSV21} for the non-relativistic case.
	 	
	 	\begin{definition}
	 	Given a constant  $\gamma>0$ we consider the two following sequences of weights. 
	 		\begin{enumerate}
	 			\item $F^{\gamma} = (F_j^{\gamma})_{j\in\Z}\in \mathcal{W}$ is defined as
	 			\begin{equation}
	 				F_j^{\gamma}(r) 
	 				:= \begin{cases}
	 					\gamma |j|^{1-\sigma}  
	 					\left(4^{-1 }\delta|j|^\sigma - r\right)_+\quad 
	 					&
	 					\mbox{if }|j|> j_0,\\
	 					0,\quad&\mbox{if }|j|\leq j_0,
	 				\end{cases} 
 				\quad\mbox{for }r>0.
	 			\end{equation}
	 			
	 			\item $G^{\gamma} = (G^{\gamma}_j)_{j\in\Z}\in \mathcal{W}$
	 			 is defined as
	 			\begin{equation}
	 				G_j^{\gamma} := \gamma\left((r/ 4c)^{1+\alpha} - \langle j \rangle)\right)_+,\quad\mbox{for }r>0.
	 			\end{equation}
	 		\end{enumerate}
	 	\end{definition}

This definition leads to the following lemma. 
	 	
	 	\begin{lemma}
	 		\label{lemma F G conditions}
 There exists $\gamma_0\in( 0,1]$
	 		such that  for all $ \gamma \leq \gamma_0 $
	 	the weights 
	 		$F^{   \gamma}$ and $G^{ \gamma}$ 
	 		verify the    conditions
	 		\eqref{F1}, \eqref{F2}, and \eqref{F3}.  
	 	\end{lemma}

	 	\begin{proof}
The proof of this lemma is based on calculations 
for analogous weights introduced originally in
\cite{CHSV21}.
Thus--in order to avoid repetition--here we only show that  Lemma \ref{lemma: Ej admissible} implies that $F^\gamma$ verifies \eqref{F2}
for a small enough $\gamma$.
More precisely, 
we take $\gamma \leq \gamma_0 = \min\{\gamma_1  , \gamma_2\}$
where
	 		\begin{equation}
	\label{eq: tunneling-constants}
	\gamma_1 
	\equiv 
  (2c)^{2\alpha}
	(
	\mu - 
	\sqrt{\mu  \mast }
	)
	(
	\mu  - \sqrt{\mu \tilde \mu}
	)
	/ 8 
	\quad \textrm{and}\quad 
	\gamma_2 
	\equiv 
	\min\Big\{ 
	\frac{\beta}{\delta \<  j_0 \>}, 1 
	\Big\} \ . 
 \end{equation}
Here, the condition $\gamma\leq \gamma_1$ is used in the proof of \eqref{F2}, 
while the condition $\gamma \leq \gamma_2$ is used in the proof of \eqref{F3}. 

\vspace{1mm}

Let us now prove \eqref{F2}. 
To this end, we note that  for  $|j| > j_0$, 
  straightforward differentiation in $r>0$
and \eqref{lemma Vj 1 eq 1}
gives 
\begin{align}
		|(F^{  \gamma}_j)'|^2 
	 \ 	=  \ 
	\gamma^2
	|j|^{2(1-\sigma)}
	\mathbb{1}_{\left[0, \delta|j|^\sigma /4 \right]}   
 \  \leq  \ 
	\gamma^2
|j|^{2(1-\sigma)} 
\Big(
\mu \delta_0  \lambda_j^2 /2 
\Big)^{-1 } 
\chi_j^\perp 
(V_j - E^2)
\chi_j^\perp   \ . 
\end{align}
Recalling that $\lambda = A^2 (c \<  j \>^\sigma)   =  c^{2\alpha} \< j\>^{2 \alpha\sigma } $, 
our  claim    now follows by taking $\gamma \leq \gamma_1$  
in the above right hand side. 
	 	\end{proof}

	 	We are now ready to prove our main result.  
	 	Since the argumentation follows the main ideas of \cite{CHSV21}, 
	 	we only present a sketch of the proof in order to avoid repetition.

	 	\begin{proof}
	 		[Proof of Theorem \ref{theorem tunneling estimates}]
	 		We choose   $\gamma= \gamma_0$ as in the proof of Lemma \ref{lemma F G conditions}, and we  denote for simplicity $ F  = F^{\gamma}$
	 		and $G = G^{\gamma}$. 

\vspace{1mm}	 		

 \textit{{(i)}}
Let us denote $ \zeta_1 \equiv \delta \gamma/8 $.
and $C_1  =  8 \delta $. 
Then   observe that for   $j \in \Z$ that satisfies $| j | > j_0$ 
it holds that  
\begin{equation}
\label{lower bound F}
 F_ j (r ) 
 \geq 
\zeta_1 \, |j|  
\qquad
\t{on}
\qquad 
	r \geq  C_1  |j|^\sigma   \ . 
\end{equation} 
Consequently, for $|j| >  j_0 $ we find that the following lower bound holds true 
\begin{equation}\label{eq 5}
	\exp(F_j) 
	\geq 
	\1_{  [ 0 ,   {C_1  |j|^\sigma ]  } } 
	\exp(F_j) 
	\geq 
	\1_{  [ 0 ,   {C_1  |j|^\sigma]  }   } 
	\exp \Big(    \zeta_1 \gamma | j |   \Big) \ . 
\end{equation}
Hence, we decompose  the sum over angular momentum  $j\in  \Z$ in two regions, 
and find that 
  \begin{align}
\nonumber 
  	 	  \sum_{ j \in \Z }
  	 	   \|  
  	 	   e^{ \zeta_1   |j |}
  	 	    			\1_{ [ 0, C_1 |j|^\sigma]   }(r )
  	 P_{m_j}
  	 \p_E (H )		\|^2 	 
  & 	 \leq 
  	 \sum_{| j | \leq j_0 } 
  	   	 	   e^{ 2 \zeta_1    |j |}
  	   	 	   + 
  	   	 	   \sum_{ |j| > j_0 }
  	   	 \|	   	\exp(F_j)   		  	 P_{m_j}
  	 \p_E (H )			\|^2 		\\ 
  & 	 \leq  
 C |j_0|   	   	 	   e^{ 2 \zeta_1     |	 j_0 	|  }
 + 
 \|	 e^{\bF} \p_E(H)	\|^2 \ . 
  \end{align}
	 		Finiteness of the right-hand side then follows from Theorem~\ref{theorem exp decay} and Lemma~\ref{lemma F G conditions}. This concludes the proof of \eqref{equation interior tunnelling}. 

	 		\vspace{1mm}

 \textit{(ii)}
 The proof of  \eqref{equation exterior tunnelling} follows a similar argument.
We set $\zeta_2 \equiv \gamma/2c $ and
$C_2 \equiv 8c$. 
	 		The only   modification comes from modifying the lower bound
	 		\eqref{lower bound F}. 
	 		In the present case, one uses the fact that for  \textit{all} $ j \in \Z $, 
	 		the inequality 
	 		$G_j  (r )  \,   \geq   \,     \,  \zeta_2  r^{1 + \alpha  }$
holds true 	 		for 
	 		$r^{1 + \alpha } \geq  C_2  \< j  \> ^\sigma $. 
	 		We leave the details to the reader. 
	 	\end{proof}

	 \section{Proof of the exponential decay theorem}
	 \label{section proof thm}
Let $H$ be the Dirac Hamiltonian satisfying Conditions \ref{condition 1'}
with $\alpha>0$ and Condition \ref{condition 2}
with    $\beta$ and $\mu_*$, 
and consider an energy level $E>0$ . 
Throughout this section, we consider 
 $\mu$, $\delta$, $c$, $j_0$, $\ve_0$, $\ve_1$, $\ve_*$, $C$ and $U$
to be   nine positive parameters, tuned with respect to $(\mu_*,  \alpha, \beta,E)$ in the sense of 
Condition \ref{condition tuning}.
In particular, the additional parameters 
$ \delta_0$, $\widetilde E$, $a$, $b$,  $\widetilde \mu$, $\bar \mu_* , \widetilde C _0$
have been defined   in \eqref{definition widetilde E}. 
  $\chi$   denotes the partition of unity  constructed with respect to $(c, \delta,j_0)$, 
  in the sense of Definition \ref{definition partition unity}.
  On the other hand $(E_j)_{j\in \Z}$
  correspond to the reference energy levels, as prescribed in  Definition \ref{definition-energies}.

\subsection{Strategy of the proof}
	 The main goal of this section is the proof of Theorem~\ref{theorem exp decay} stated in the previous section. Let us first explain the general ingredients: As in  \cite{BFS1998b,Griesemer2004} (see also \cite{CHSV21}) we use the Helffer-Sj\"ostrand formula to approximate  the spectral projection  
	 $\p_I (H)=\mathds{1}_{ [-E , E]} (H)$ by a smooth function  $g(H)$ of the Hamiltonian. Let  
	 $g \in C_c^\infty (\R, [0,1])$
	 be such that $g |_{ [-E,E]   } = 1$ and $\supp \,  g \subset [ -E - \ve/2 , E + \ve  /2 ] $, for some $\ve>0$. By means of the Spectral Theorem we see that it is enough to show the boundedness of $e^{\bF}g(H)$ since
	 \begin{equation}
	 	\|e^{\bF } \p_I(H)\| = \|e^{\bF } g(H)\p_I(H )\| \leq \|e^{\bF } g(H)\|\,.
	 \end{equation}
	 
	 It is well-known (see e.g., \cite{Davies1995}) that there is an almost analytic extension of $g$, denoted by $\tilde{  g  }$, such that 
	 \begin{equation}\label{hs}
	 	g (t) =   \frac{1}{ \pi   } \int_{\triangle} 
	 	( t - z )^{-1} \,\partial_{\bar z} \, 
	 	\tilde{g}(z)\, \d x \d y  \,,\qquad t\in \R .
	 \end{equation}
	 Here $\partial_{\bar z} \tilde{g }=\frac{1}{2}(\partial_x + {\im} \partial_y)\tilde{g }$ is supported on 
	 $ \triangle  := \{ z = x+iy \in \C \, \  |  \   |x| \leq E + \ve , \ |y | \leq \ve \}$. The almost analytical property allow us to pick $\tilde{  g  }$ such that 
	 $ |\partial_{\bar z} \tilde{g }(z) |   = \mathcal{O} \big(   | \mathrm{Im} z|  \big) $ as $ | \mathrm{ Im} z| \downarrow  0 $.  To simplify the notation we choose $\ve=\ve_0$, the localization error from our previous analysis.

	 \vspace{1mm}

	 It is useful to compare $H$ to a suitable self-adjoint operator $\widetilde{H}$ such that the resolvent set of the latter is contained in $\triangle$, 
	 \textit{i.e.} $\rho(\widetilde{H})\subset \triangle$. 
	 We refer to the latter property as a \textit{spectral gap} condition. 
Indeed, in this case $g(\widetilde H) = 0 $
and the following equation holds 
	 \begin{equation}\label{g-t}
	 	g (H )  
	 	=
	 	\frac{1}{\pi}\int_{\Delta}
	 	\Big( \big(H - z\big)^{-1} -  \big(  \widetilde{H} - z\big)^{-1}
	 	\Big) 
	 	\partial_{\bar z} \,  \widetilde{ g }  (z) \, \d x \d y  \, . 
	 \end{equation}
	  Theorem \ref{theorem exp decay} 
	  might then be shown by getting suitable estimates on the operator norm of  the composition 
	  $e^{\bF }  \big[ (H - z)^{-1} -  (  \widetilde{H} - z)^{-1}\big]$.
	 
	 \vspace{1mm}

	 An essential novel ingredient in our approach is in the choice of the operator $\widetilde{H}$. 
	 Indeed, we work with a comparison operator   acting as $H$,  plus an auxiliary mass-like potential.
Formally, this is defined as the sum 
	 \begin{align}
	 	\widetilde{H} = H+\sigma_3\bE \bchi	 
	 \end{align}
	 where	   $\bE \bchi=\oplus_{ j \in \Z } E_j\chi_j$ and  $\bE$ is chosen as in Definition~\ref{definition-energies}.  
	  Note that $\bE$ is an unbounded operator in case $W$ grows at infinity (and hence $\mast>0$). 
	  At least in the case when $W=0$ it is is easy to check that the desired spectral gap holds.

	 \vspace{1mm}
	 
	 In the next subsections we construct   $\widetilde{H}$ as a self-adjoint operator   and provide, in Proposition~\ref{prop inverse estimate},  the main bounds to 
	 estimate $e^{\bF } $ times the resolvent difference.
	 We give  the proof of Theorem~\ref{theorem exp decay}
	 at the end of this section.

	 \subsection{The free comparison operator}
	 We consider first the \textit{free} comparison  operator  acting initially on $\core$ as,
	 \begin{align}
	 	\label{h0-tilde definition}
	 	\widetilde H _0 
	 	: =
	 	H_0  
	 	+ 
	 	\sigma_3 {\bE}\bchi
	 	\quad \textrm{where} \quad  
	 	{\bE}\bchi 
	 	\equiv \oplus_{j\in\Z} 
	 	E_j \chi_j \ . 
	 \end{align} 
	 In   Lemma~\ref{lemma-essential}, we show  that $\widetilde H_0$ is essentially self-adjoint.
	 We denote its closure by the same symbol $\widetilde H_0$, 
	 and its  domain by
	 $\calD (\widetilde H _0)$. 
	 We  characterize the latter in Proposition~\ref{proposition domain comparison} below.

	 A short calculation yields	that as quadratic forms on $\core$
	 \begin{equation}
	 	\label{dd1}
	 	\widetilde{H}_0^2
	 	\ = \ 
	 	H_0^2 + 
	  \chi	{\bE}^2  	\chi 
	 	+L({\bE})\,,
	 \end{equation}		
	 where
	 $L({\bE}) \equiv  \oplus_{ j \in \Z } E_j \chi_j' \, \sigma_1\,.$
	 Note that while the second term in \eqref{dd1}  is quadratic in $\bE$, 
	 the third term is only linear. 
Hence, it can then be regarded as a lower order term for large energies.

%
%
	 \begin{lemma} \label{lemma: H tilde gap}
	 	In the sense of quadratic forms on $\core$, we have
	 	\begin{equation}\begin{split}\label{dd3}
	 			H_0^2 + 
 \bchi 	 			{\bE}^2    	\bchi 
	 			\ge ({\bE}^2-2\varepsilon_0)\,.
	 	\end{split}\end{equation}
Furthermore, the following statements hold true. 
	 	\begin{enumerate}
	 		\item $\calD
	 		( \widetilde H _0 )
	 		\subset   
	 		\calD (  {\bE}  )$ and  
	 		the following estimate holds 
	 		\begin{equation}\begin{split}\label{lower bound m^2}
	 				\|  \widetilde{H}_0 \vp   \| 		\textstyle
	 				\geq   (1-
	 			 3 \ve_0
	 			 )^{\frac{1}{2}}
	 				\| {\bE}\, \vp   	\|
	 				\ , 
	 				\qquad 
	 				\forall \vp   \in \calD ( \widetilde H _0  ) \ . 
	 		\end{split}\end{equation} 
	 		\item $\calD
	 		( \widetilde H _0 )
	 		\subset   
	 		\calD (  H_0  )$ and  
	 		\begin{equation}\begin{split}\label{lower bound two}
	 				\|  \widetilde{H}_0 \vp   \| 		\textstyle
	 				\geq 
	 				 (1 - 2\ve_0)^{\frac{1}{2}}
	 				\| H_0\, \vp   	\|
	 				\ , 
	 				\qquad 
	 				\forall \vp   \in \calD ( \widetilde H _0  ) \ . 
	 		\end{split}\end{equation} 
	 	\end{enumerate}
	 \end{lemma}
	 \begin{proof}
	 	Let us first show \eqref{dd3}. 
	 	We observe that $H_0^2+{\bE}\bchi^2=\oplus_{ j \in \Z }(h_j^2+E_j^2\chi^2_j)$ and verify that
\begin{align}
\nonumber 
	\boldsymbol{h}^2 + \bchi \bE^2 \bchi 
& 	 \  =  \ 
	 T + \bV + \bchi \bE^2 \bchi 	\\ 
\nonumber 
	&   \  =  \ 
	   T  + \bchi \bV \bchi + \bchi^\perp( \bV - \bE^2 ) \bchi^\perp + \bE^2  	\\ 
 & 	 \    \geq  \ 
	    \bchi 	\boldsymbol{h}^2 \bchi + \bchi^\perp T \bchi^\perp +   
	    \bchi^\perp (\bV - \bE^2 ) \bchi^\perp 
	    + \bE^2 - 2\ve_0 \ . 
\end{align}
	 	The claim now follows by using  the positivity of $T$ and $h_j^2$ as well as \eqref{eq: Ej admissible}.
	 		 	In order to show \eqref{lower bound m^2} on $\core$, 
	 		 	one may verify that 
	 	\begin{equation}\label{lerror}
	 		\pm L (\bE)
	 		 \  \leq 	\		|\bchi '| \bE 
	 		  \ \le \ 
	 		   \ve_0 \bE
	 		   \ \le  \ 
	 		\ve_0 \bE^2 
	 	\end{equation}
where in the last inequality  we used $\bE \geq1$. 
This, 	 	   combined with \eqref{dd1} and  \eqref{dd3}, readily implies  
	 	\begin{equation}
	 		\widetilde{H}_0^2
	 		\ \ge\ 
	 		 (1-3 \ve_0)\bE^2 \ . 
	 	\end{equation}
	 	This yields \eqref{lower bound m^2} on $\core$, 
	 	and then on $\calD (\widetilde H _0)$ after a density argument. 
Finally, we prove 	 	 	\eqref{lower bound two}.
To this end,  we write
using 
\eqref{dd1},  \eqref{dd3} and \eqref{lerror} 
\begin{align}
\nonumber 
	\widetilde H _0^2  
& 	\geq
 H_0^2 + \bchi \bE^2 \bchi - \ve_0  \bE^2   \\ 
\nonumber 
& \geq
 (1-2\ve_0) H_0^2 
+
 2 \ve_0 (H_0^2 + \bchi \bE^2  \bchi )	- \ve_0 \bE^2 \\ 
& \geq 
(1 - 2 \ve_0 )H_0^2 
+
  \big[ 2 \ve_0 (\bE^2  -  2 \ve_0) - \ve_0 \bE^2  \big ] \ . 
\end{align}
Finally, we note that the term in brackets $[\cdots]$ in the last displayed equation is non-negative.
This follows from the lower bound $\bE \geq 1 $
and the fact that $\ve_0 \leq 1/4 $.  This finishes the proof. 
	 \end{proof}
	 Next we characterize the domain of self-adjointness of the free comparison operator as follows. 
	 \begin{proposition}
	 	\label{proposition domain comparison}
	 	We have  
	 	\begin{equation}
	 		\textstyle 	\calD (\widetilde H _0) =  \mathcal{D}(H_0)\cap \mathcal{D}({\bE})=
	 		\big\{
	 		\, 
	 		\vp \in \calD( H_0 )
	 		\ |  \ 
	 		\textstyle  \sum_{j \in \Z } E_j^2  \,  \| P_{m_j}  \vp   \|^2 < \infty 
	 		\, 
	 		\big\} \ .
	 	\end{equation}
	 \end{proposition}
	 \begin{proof}
	 	We prove in Lemma B.3 
	 	that    $\widetilde H _0 |_{\calC}$ is essentially
	 	self-adjoint. 
	 	From Lemma~\ref{lemma: H tilde gap}
	 	we see that
	 	$\calD (\widetilde H_0) \subset \calD (\bE) \cap \calD (H_0)$.
	 	Further, for $\vp\in\calC $
	 	one clearly has    that
	 	\begin{align}\label{in}
	 		\| \widetilde H _0 \vp \|^2 
	 		\leq 
	 		2\left(\|   H_0 \vp\|^2+ 
	 		\|  \bE \vp  \|^2+\|\vp  \|^2\right)
	 		= 
	 		2\left(\|  \sqrt{ H_0^2+\bE^2}\, \vp\|^2+ 
	 		\|\vp  \|^2\right)
	 		\,.
	 	\end{align}
	 	Applying the spectral theorem for the strongly commuting operators $H_0$ and $\bE$  (see, e.g., \cite[Sec. 5.5]{schmudgen2012unbounded}) it is easy to see that
	 	$\mathcal{D}(H_0)\cap \mathcal{D}(\bE)$ equals the domain of the closure of ${\sqrt{ H_0^2+\bE^2}\!\upharpoonright_\core}$. Hence, we conclude that 
	 	$\calD (H_0) \cap \calD (\bE) \subset \calD (\widetilde H _0)$.
	 	This finishes the proof. 
	 \end{proof}

	 \subsection{The full comparison operator}
Let us now  start our discussion of the  perturbations of the free comparison operator. 
First, we prove the following technical lemma.
It shows that  under assumptions on the derivatives of the exponential weights, 
the operators $\exp(\pm F)$
continuously map the operator domain of $H_0$  into itself.

	 		\begin{lemma}\label{lemma: eFH domain}
	 			Let    $F \in  \W_b   $ be a bounded weight satisfying    \eqref{F2}. 
	 			Then, the following statements hold true 
	 			\begin{enumerate}[label = (\roman*)]
	 				\item $e^{ \pm  \bF }  \calD(H_0)  = \calD(H_0)$. 
	 				
	 				\item 
	 					There exists 
	 				$C>0$ such that 
	 				for all $\vp \in \calD(H_0 )$
	 				\begin{equation}
	 					\|	 H_0 e^{\bF} \vp	\| \leq C \|e^{\bF }	\| 
	 					\| (H_0 +  i )	 \vp\|	  \ . 
	 				\end{equation}
	 			\end{enumerate}
	 		 
	 		\end{lemma}

	 		\begin{proof}
	 				 			Let us fix $ j \in \Z$. 
	 			Our first observation is the following inequality, 
	 			holding in the sense of quadratic forms of $\Cr$ 
	 			\begin{equation}
	 				\label{lemma F ineq 1}
	 				|F_j'|^
	 				2
	 				\ \leq \ 
	 				\theta  
	 				\, 
	 				\chi_j^\perp	 	 \, 			  V_j				\, 	\chi_j^\perp	 
	 				\	\leq 	\ 
	 				\theta  \,  (h_j^2 + 2\ve_0^2) 
	 			\end{equation}
	 			where  $\theta =   (1/4)
	 			( 1 -      a 	)^{1/2}$
	 			is the constant from \eqref{F2}, 
	 			and we have made use of 
	 			\eqref{v-for-h}. 
Since $\Cr$ is a domain core for both $h_j$ and $F_j'$, we conclude that  
$\calD(h_j) \subset \calD (F_j ' ) $
and \eqref{lemma F ineq 1} extends to $\calD(h_j)$ 
in the form $ \| F_j ' \vp 	\| \leq C  ( \|   h_j \vp		\|  + \|\vp	\|) $, for $\vp \in \calD(h_j)$. 

\vspace{1mm}

Let us now prove \textit{(i)}. To this end, let us fix  $ \psi \in \calD(h_j)$.  
	 			We claim  that $e^{\pm F_j} \psi \in \calD(h_j^*)$. 
	 			To this end, we 
	 			recall the definition of $h_j$ in 
	 			\eqref{Dj definition}  and   compute that 
	 			as linear operators on $\Cr$
	 			\begin{equation} 
	 				\label{commutator id 1}
	 				[ h_j   ,   e^{ \pm F_j}]    =  
	 				e^{  \pm  F_j} 
	 				\begin{pmatrix}
	 					0 &     \mp  F_j '  \\
	 					\pm  F_j ' & 0 
	 				\end{pmatrix} \, 
	 				=  \pm 
	 				\im e^{ - F_j}    \sigma_2F_j'   
	 			\end{equation}
	 			where $\sigma_2$ is the second Pauli matrix.  Then, for any  $\zeta \in \Cr$ it follows that 
	 			\begin{equation}
	 				\<  e^{\pm F_j} \psi  , h_j \zeta   \>
	 				=
	 				\< 
	 				\psi 
	 				, 
	 				e^{ \pm F_j}
	 				h_j \zeta 
	 				\>
	 				=
	 				\<
	 				\psi 
	 				,
	 				h_j e^{ \pm F_j} \zeta
	 				\pm 
	 				\im e^{ \pm F_j }\sigma_2 F_j ' \zeta 
	 				\>
	 				=
	 				\< 
	 				e^{\pm F_j } (h_j \mp \im \sigma_2 F_j' ) \psi  , \zeta 
	 				\> \ . 
	 			\end{equation}
Thus, our claim now follows after we show that $ 	e^{\pm F_j } (h_j \mp \im \sigma_2 F_j' ) \psi \in L^2(\R_+)$. 
This follows from the fact that $\exp(\pm F_j)$ are bounded operators, 
from our observation 
$\calD(h_j) \subset \calD (F_j ' ) $, 
and  the original assumption $\psi \in \calD(h_j)$. 
Since $h_j$ is self-adjoint, this implies that $\exp(\pm F_j) \calD(h_j) \subset \calD(h_j)$. 
Thus, the first claim of the lemma is now proven once we sum over angular momentum channels $j\in \Z$, 
and use the fact that $\exp(\pm F)$ are continuous inverses of one another. 

\vspace{1mm}
The proof of \textit{(ii)}  
now follows from the commutator identity 
\eqref{commutator id 1}, 
inequality
\eqref{lemma F ineq 1}, 
and an approximation argument using the fact that 
$\calC$
is a domain core of $H_0$. 
Namely, we see that  for $\vp \in \calD (H_0)$ and $\zeta \in \calC$
\begin{equation}
	\<  \zeta, H_0 e^{\bF } \vp		\>
	 = 
	 \<	 e^{\bF} H_0\zeta, \vp	 \>
	  = 
	   \<  (  H_0 e^{\bF}    - i \sigma_2 \nabla {\bF}  e^{\bF} ) \zeta, \vp	\>
	    =
	    \<  e^{\bF} \zeta,  (  H_0 - i \sigma_2 \nabla \bF )  \vp\>  \ . 
\end{equation}
It suffices now to use \eqref{lemma F ineq 1} over all $  j \in \Z $, 
and take the supremum over $\zeta \in \calC$, keeping in mind that $e^{\bF}$ is a bounded operator.
	 		\end{proof}

Next, we turn to the following lemma. 
With it,  we can control the perturbation $W$  in terms of the 
free comparison operator. 

\begin{lemma}
	\label{lemma W control}
	For all $\vp  \in \calC $  the following inequality holds true 
	\begin{equation}
	 		\label{W-H-tilde}
		\|	 W \vp 	\|^2	
			\leq  
\Big(  
	  	\frac{	\widetilde \mu 	}{\mu}
\Big)^{\frac{1}{2}}  
	\|  \widetilde H _0 \vp 	\|^2    \ . 
	\end{equation} 
In particular, $\calD(\widetilde H _0) \subset \calD(W)$
and the inequality \eqref{W-H-tilde} extends to $\vp  \in \calD(\widetilde H_0). $
\end{lemma} 

Let us briefly postpone the proof of the above lemma, 
and introduce  the \textit{full} comparison operator  as follows 
\begin{equation}
	\label{H tilde definition}
	\widetilde H \equiv \widetilde H_0 + W 
	\qquad
	\t{on}
	\qquad 
 \calD (\widetilde H )  \equiv 	\calD(\widetilde H_0) \ . 
\end{equation}
In particular, thanks to Lemma \ref{lemma W control}, inequality \eqref{p3}, and      the Kato-Rellich theorem, 
$\widetilde H $ is  as a self-adjoint operator.  
Furthermore, the space $\calC$ introduced  in Definition \ref{definition domain cores}
is a domain core for $\widetilde H $. 

\vspace{1mm}

In the next proposition, we establish fundamental estimates regarding $\widetilde H$, when intertwined with an exponential weight
$F \in \W_b$, satisfying the conditions \eqref{F1}, \eqref{F2} and \eqref{F3}.  
Let us  note that thanks to    Lemma
\ref{lemma: eFH domain}  we may easily 
conclude that   $\exp(\pm F) \calD( \widetilde H_0) = \calD(\widetilde H_0)$.
 Thus, we define the 
 exponentially twisted operator  as 
 \begin{equation}
\widetilde H^{F}
\equiv e^{\bF} 
\widetilde H 
e^{-\bF } 
\qquad
\t{on}
\qquad 
\calD (\widetilde H^F) \equiv   \calD(\widetilde H_0) \  . 
 \end{equation}
Let us note that the operator $\widetilde H^F$  is closed.
Indeed, 
since $\widetilde  H $ is self-adjoint and  the maps $\exp(\pm \bF )$ are continuous, self-adjoint  and  inverses of one other, 
one may  easily check that 
$(\widetilde H ^F)^* = \widetilde H^{-F}$.
In particular, this shows that 
$ \overline{\widetilde H^F} = (\widetilde H ^F)^{ * *} = \widetilde H^F$, as desired. 
Furthermore, the inequality contained in \eqref{lemma: eFH domain}
easily implies that $\calC$ is also a domain core of  $\widetilde H^F$.

	 \begin{proposition}
	 	\label{prop inverse estimate}
The following statements hold true. 
\begin{enumerate}[label=(\roman*), leftmargin  =1 cm]
	\item 
For  all 
$\psi \in \calD (\widetilde H )$
 there holds 
	\begin{align}
		\label{comparison op gap}
		\|\widetilde{H}\psi\|  \,\geq\,  \auxE		\|\psi\| \ . 
	\end{align}
	
	\item Additionally, 
assume that 	$F    \in \mathcal{W}_b$ satisfies  \eqref{F2} and \eqref{F3}.  
	Then, for all  
	$\psi \in \calD (\widetilde H^F  )$
	 and $z\in\C$
	 with 
 $|z| \leq U/4$ there holds 
	\begin{equation}
		\begin{split}
			\label{twisted full comparison op lower bound}
			\|\big
			( \widetilde H^F 	-			z\big)
			\psi\| 
			\,\geq\, 
		 (\delta_0 /4 ) 
			\|\mathbf{E}\psi\|  \ . 
	\end{split}\end{equation}
\end{enumerate}
 \end{proposition}

 Now, we turn to the proof of Lemma \ref{lemma W control}, 
 and Proposition \ref{prop inverse estimate}. 
 
 \begin{proof}[Proof of Lemma \ref{lemma W control}]
 		   Let us recall some    important  notations
 	that will make the proof easier to navigate.
 	Namely,  in  Condition \ref{condition tuning}
 	we have introduced 
 	\begin{equation}
 		a  
 		=   
 	(		\widetilde \mu / \mu )^{\frac{1}{2}}  , 
 	\qquad
 	b 
 	 =  
 	( 	\overline{\mu}_*	 /   \widetilde \mu )^\frac{1}{2} , 
 	\qquad 
 	\t{and}
 	\qquad
 	\delto  
 	=  
 	\min(
 	1 - a, 1 - b
 	)  
 	\end{equation}
 	In particular, in this notation we may write the reference energy levels   (see Definition~\ref{definition-energies})
 	in the following way 
 	\begin{equation}
 		\bE^2
 		=
 		\tfrac{b}{a}\tilde{\mu} {\boldsymbol \lambda}^2
 		+
 		\delta_0^{-2 }U^2 \ . 
 	\end{equation}
 	with $\widetilde \mu < 1 $ and $U \geq 1  $. 
\vspace{1mm }

 	Let us now prove 
 		 		\eqref{W-H-tilde}. 
 To this end,  we use 
 	Theorem \ref{theorem W}
 	and 
 	the error bound 
 	\eqref{lerror}
 	to find 
 	that  in the sense of quadratic forms on $\mathcal{C}$ 
 	\begin{align}
\label{W bound}
 		\widetilde H_0 ^2
 		\  \geq  \ H_0^2 + \bchi \bE^2 \bchi   - \ve_0 \bE 
 		\qquad 
 		\t{and}
 		\qquad
 		\frac{1}{a}  W^2 
 		\ 	\leq  
 		\ 
 		a H_0^2 + b \bchi \bE^2 \bchi + \frac{C_0}{a} \ .
 	\end{align}
 	Thus, the reverse triangle inequality, 
 	the definition of $\delta_0$  
 	and the lower bound
 	\eqref{dd3} imply 
 	\begin{align}
 		\label{H_0 - v_beta}
 		\widetilde{H}_0^2 - \frac{1}{a} W^2 
 		&  \ \geq  \ 
 		(1-a)  H_0^2 +  (1-b) 
 		\chi \bE^2  \chi
 		- \ve_0\bE -\frac{C_0}{a} \\ 
 		& \ \geq \ 
 		\delta_0 
 		(					H_0^2 + \bchi \bE^2 \bchi 			)	
 		- \ve_0\bE -\frac{C_0}{a}   \\ 
 		& \ \geq \  
 		\delta_0(\bE^2-2\ve_0) -\ve_0\bE-\frac{C_0}{a}.
 		\label{ineq H w 1 }
 	\end{align}
In view of the inequalities  $\ve_0   \bE   \leq \ve_0 \bE^2, \ $ 
$\bE \geq \delta_0^{-1} U ,\ $ 
\eqref{p4}
and   \eqref{definition widetilde E}, the above right hand side is positive. 
Hence, 
$W^2 \leq a \widetilde H_0 ^2$ in the sense of quadratic forms on $\calC$.  
The last proof of the lemma follows from the fact that
$\calC$ is a domain core for both $W$ and $\widetilde H_0$. 
 \end{proof}

\begin{proof}
	[Proof of Proposition \ref{prop inverse estimate}]
Let $a,$ $b$ and $\delta_0$ as in the beginning of the proof of Lemma \ref{lemma W control}.	 

\vspace{1mm}

 \textit{(i)}	We start by showing \eqref{comparison op gap}. 
To this end, let $\psi \in \core$ be in the domain core as defined in Definition \ref{definition domain cores}. 
Then, 	 we make use of the reverse triangle inequality
and 
\eqref{ineq H w 1 }  
to find that 
\begin{equation}\begin{split}\label{H tilde triangle ineq}
		\|(\widetilde{H}_0 +  W  )\psi\|^2 
		&\geq (1-a)\|\widetilde{H}_0\psi\|^2 + \left(1 - \frac{1}{a}\right)\| W  \psi\|^2
		\\
		&  =  \, (1-a)\left(\|\widetilde{H}_0\psi\|^2 - \frac{1}{a}\| W \psi\|^2\right)\,   \\ 
		& \geq 
		(1-a )
		\<		\psi ,  
		\Big[ 
		 		\delta_0(\bE^2-2\ve_0) -\ve_0\bE-\frac{C_0}{a}
		\Big]
		 \psi 	 \> \ .
\end{split}
\end{equation} 
In view of the inequalities  $\ve_0   \bE   \leq \ve_0 \bE^2, \ $ 
$\bE \geq \delta_0^{-1} U ,\ $ 
\eqref{p4}
and   \eqref{definition widetilde E}, 
the term in $[ \cdots]$  is bounded below by 
$\delta_0^{-1 } \widetilde E^2$. 
Since $1 -a \geq \delta_0$, 
this implies that 
\begin{equation}
			\|(\widetilde{H}_0 +  W  )\psi\|^2  
			\geq 
			 \widetilde E ^2
			  \|	 \psi	\|^2 \ . 
\end{equation}
Since $\calC$ is a domain core for $\widetilde H = \widetilde H _0 + W $, 
the above inequality extends to $\psi \in \calD(\widetilde H)$. 

\vspace{1mm}
	 \textit{(ii)}	We now turn to the proof of \eqref{twisted full comparison op lower bound}.  Thanks to Lemma~\ref{lemma: eFH domain} we have that, for all $\psi \in \core$,
\begin{equation}\label{aqa}
	e^{\bF} 			
	(   \widetilde{H}_0   + W ) 
	e^{-\bF}  \psi \, = \, 
	\big( 
	\widetilde{H}_0  + e^{\bF } We^{-\bF} - \im \sg\cdot\nabla \bF\big) \psi\,. 
\end{equation}
Hence, using Lemma \ref{lemma exp twisted}
and   the same idea as in \eqref{H tilde triangle ineq} we get
\begin{equation}\begin{split}\label{twisted H tilde triangle ineq}
		\| e^{\bF} 			
		(   \widetilde{H}_0   + W ) 
		e^{-\bF}  \psi    \|^2
		& \, \geq  \, 
		\frac{1}{2}\|    (\widetilde{H}_0 + e^{\bF } We^{-\bF}) \psi   \|^2							
		-
		2\|  \nabla \bF \psi    \|^2
		\\
		&\,\geq\, \frac{1-a}{2}\left(\|\widetilde{H}_0\psi\|^2 - \frac{1}{a}\|v_\beta\psi\|^2 - \frac{4}{1-a}\|\nabla \bF\psi\|^2\right).
\end{split}\end{equation}
Moreover, in view of 	 			\eqref{v-for-h} and \eqref{F2} we have that, in the sense of quadratic forms on $\core$,
\begin{equation}\label{eq: nabla F - H0}
	\begin{split}
		|\nabla F|^2\leq \frac{(1-a)^2}{4}\big(H_0^2 + 2\ve_0 - \bchi^\perp  \bE^2  \bchi^\perp  \big).
	\end{split}
\end{equation}
Note that the inequality for $W$ in \eqref{W bound} also holds for $v_\beta$--see Theorem \ref{theorem W}. 
Hence, combining \eqref{H_0 - v_beta} and \eqref{eq: nabla F - H0} in \eqref{twisted H tilde triangle ineq} gives
\begin{align}
	\nonumber 
	\| e^{\bF} 			
	(   &\widetilde{H}_0   + W ) 
	e^{-\bF}  \psi    \|^2							\\ 
	\nonumber  
	&\geq\, \frac{1-a}{2}\left\{a \|H_0\psi\|^2 + (1-b)\|\bE\bchi\psi\|^2+(1-a)
	\|\bE\chi^\perp\psi\|^2
	- \ve_0\|\bE^{\frac{1}{2}}\psi\|^2
	-
	\tilde{C}_0\|\psi\|^2\right\}\\
	&\geq\, \frac{\delta_0}{2}\left\{a \|H_0\psi\|^2+\delta_0\|\bE\psi\|^2 - \ve_0
	\|		\bE^{\frac{1}{2}}   \psi\|^2 -
	\tilde{C}_0\|\psi\|^2\right\}\,,
\end{align}
where $\tilde{C}_0 \equiv 2(1-a)\ve_0 +  C_0/a  $.
Note that  $ \ve_0 \bE \leq \ve_0 \bE^2 $, 
$\bE \geq  \delta_0^{-1} U $
and $U$ is chosen large   relative to $C_0$ 
in \eqref{p4}.
Hence,  
the last two terms can be dropped  
at the expense of taking $\delta_0/2$ away from the second term, 
\textit{i.e.} 
$ ( \delta_0/2) \bE^2 - \ve_0\bE-  \tilde{C}_0 \geq 0 $. 
Hence, we find

\begin{equation}\label{coers}
	\|e^{\bF}	\widetilde{H}	e^{-\bF}	\psi\|^2
	 \ \geq  \ 
	\frac{\delta_0 a}{2} \|H_0\psi\|^2 +\frac{\delta_0^2}{4}\|\bE\psi\|^2.
\end{equation}
Finally, we use the
lower bound \eqref{coers}
and  the 	 	 reverse triangle inequality 
to find that (here, we drop the $H_0$ term)
\begin{align}
\nonumber 
	\|	 (	 e^{\bF} \widetilde H e^{-\bF}  - z			)		 \psi	\| 
&  \  	\geq  \ 
	\|	 e^{\bF} \widetilde H e^{-\bF}   \psi	\| 
	- 
	|z |
	\|	 \psi	\| 	\\
	\nonumber 
&  \ 	\geq  \ 
	( 	  \delta_0/2)  \|	 \bE \psi	\| 
	- |z |		\|	 \psi	\| 	\\ 
&  \ 	\geq  \ 
	(  \delta_0/4)  \|	 \bE \psi	\|  \ . 
		\label{eq: no z inv norm}
\end{align}
Here, we have used the fact that $\bE \geq \delta_0^{-1} U$
combined with the assumption 
$|z| \leq U/4$. 
This finishes the proof  for $\psi \in \calC $. 
The inequality then extends to $\psi \in \calD(\widetilde H^F	)$ 
since $\calC$ is a domain core
for both $\widetilde H^F$ and
$\bE$. 
\end{proof} 
	 
	 \subsection{Proof of Theorem~\ref{theorem exp decay}}
	 
The next lemma contains one of the most important bounds
concerning the exponentially intertwined operator $\widetilde H^F$.

	 \begin{lemma}\label{lemma: HF extension}
	 
	 	Let  $F   \in \mathcal{W}_b$ satisfy \eqref{F2} and \eqref{F3}. 
	 	Then, the following statements hold true. 
	 	\begin{enumerate}[label = (\roman*)]
	 		\item For any $z\in \rho(\widetilde{H}^{F})\cap\rho(\widetilde{H})$, we have 
	 		\begin{align}\label{ri}
	 			(\widetilde{H}^{ F}-z)^{-1}=e^{ \bF}(\widetilde{H}-z)^{-1} e^{- \bF}\,,
	 		\end{align} 
	 		
	 		\item 
	 		 For  any $| z|\le U /4 $ we have that $z\in \rho(\widetilde{H}^{F})$ and, 
	 		for all $\psi\in \H$, 
	 		\begin{equation}\label{bound2}
	 			\|\bE(\widetilde{H}^{F}-z)^{-1}\psi\| 
	 			 \ \leq  \ 
	 		  (4 /\delta_0 )
	 			 	\|\psi\| 			 \ . 
	 		\end{equation}
	 	\end{enumerate}
	 \end{lemma}
 \begin{proof}
First we prove \textit{(i)}.
To this end,  we observe that for  
$z\in \rho(\widetilde{H}^{F})\cap\rho(\widetilde{H})$ 
and
 $\vp\in\calD(\widetilde H ) = \calD (\widetilde H ^F )$ we have
	\begin{equation}\label{eq: eF to HF}
		e^{\bF}
		(\widetilde{H}-z)^{-1}e^{- \bF}(\widetilde{H}^{ F}-z)\vp = \vp
		=(\widetilde{H}^{ F}-z)^{-1}(\widetilde{H}^{ F}-z)\vp\,.
	\end{equation}
	This implies the 
	operator equality \eqref{ri}. 
	
Now, we prove \textit{(ii)}. 	  
Our starting point is the lower bound  
	 	\eqref{twisted full comparison op lower bound}, which we recall reads  
	 	\begin{equation} \label{extnorm}
	 		\|(\widetilde{H}^{ F}-z)\vp\| 
	 		 \ \geq  \ 
(  \delta_0/4 ) 	\,  \|\bE \vp\| 
	 		 \qquad \forall \vp\in
	 		\calD (\widetilde H ^F)  \  .
	 	\end{equation}
	 	Next, we show that $z\in \rho(\widetilde{H}^{ F})$ 
	 	provided $| z| \le  U /4 $. Equation \eqref{extnorm} implies that $\widetilde{H}^F$ is invertible with  continuous inverse from 
	 	$\mathrm{Ran} ( \widetilde{H}^{F}-z ) $
	 	 to $\mathcal{H}$. 
	 	Noting that \eqref{extnorm} holds also for  $\widetilde{H}^F$  and $z$ replaced by $\widetilde{H}^{-F}$ and $z^\ast$, respectively, we see that $\mathrm{Ker}(\widetilde{H}^{-F}-z^\ast)=
	 	\mathrm{Ker}(\widetilde{H}^{F}-z)^\ast=\{0\}$. 
	 	Therefore, $\mathrm{Ran} ( \widetilde{H}^{F}-z ) $ is dense. 
	 	Finally, note that by the continuity  of 
	 	$(\widetilde{H}^{F}-z)^{-1}$ and the Closed Graph Theorem we find that 
	 	the  range of $\widetilde{H}^{F}-z$ equals $\mathcal{H}$, i.e., $z\in \rho(\widetilde{H}^{ F})$. Using the bijectivity of the resolvent in \eqref{extnorm}
	 	we get \eqref{bound2}.
	 	This finishes the proof. 
	 		 \end{proof}
	 

	 We are now ready to prove the abstract exponential decay theorem that
	 we introduced in the last section. This then completes the proof of our main result. 
	 
	 \begin{proof}[Proof of Theorem~\ref{theorem exp decay}]
	 	Let us first assume that $F \in \W_b   $ conditions
	 	 \eqref{F1}, \eqref{F2} and \eqref{F3}. 
	 	We consider 
	 	$\widetilde{H}$ as in \eqref{H tilde definition}
	 	and \eqref{h0-tilde definition}. 
	 	For Lemma~\ref{lemma: HF extension} and Proposition~\ref{prop inverse estimate} to hold we choose 
	 	$U>0$
	 	large enough in Condition \ref{condition tuning}  so that 
	 	$\sup_{ z \in \triangle} |z |  \leq U/4  $.

	 	We follow the construction described at the introduction of this section. Note that thanks to \eqref{W-H-tilde} we verify that $\rho(\widetilde{H})\subset \triangle$. Hence, formula \eqref{g-t} applies. 
	 	
	 	We study the norm of  
	 	\begin{align}
	 		e^{\bF}\left((H-z)^{-1} - (\widetilde{H}-z)^{-1}\right), \quad\mbox{for}\quad z\in \triangle\quad \mbox{a.e}\,.
	 	\end{align}
	 	Since $\calD(\widetilde{H})\subset\calD(\bE)$ we have that $(H-z)^{-1} - (\widetilde{H}-z)^{-1} = (H-z)^{-1}(\widetilde{H}-H)(\widetilde{H}-z)^{-1}$. 
	 	For any $\vp,\phi\in \H$ we look
	 	at
	 	\begin{equation}
	 		\begin{split}
	 			\langle e^{\bF} \left((H-z)^{-1} - (\widetilde{H}-z)^{-1}\right) \vp, \phi \rangle
	 			&=\langle \vp, \left((H-z^\ast)^{-1} - (\widetilde{H}-z^\ast)^{-1}\right)e^{\bF}\phi \rangle\\
	 			&= \langle (H-z)^{-1}\vp,\, (\widetilde{H}-H)(\widetilde{H}-z^\ast)^{-1}e^{\bF}\phi \rangle\,.
	 		\end{split}
	 	\end{equation}
	 	Therefore, 
	 	\begin{equation}\label{griesemer norms}
	 		\begin{split}
	 			\|e^{\bF}\left((H-z)^{-1} - (\widetilde{H}-z)^{-1}\right)\|&\leq \|(H-z)^{-1}\|\|(\widetilde{H}-H)(\widetilde{H}-z^\ast)^{-1}e^{\bF}\|\\
	 			&= \|(H-z)^{-1}\| \|\bE{ \bchi} e^{\bF} e^{-\bF}(\widetilde{H}-z^\ast)^{-1}e^{\bF}\|\\
	 			&\leq \|(H-z)^{-1}\| \|\bE e^{-\bF}(\widetilde{H}-z^\ast)^{-1}e^{\bF}\|,
	 		\end{split}
	 	\end{equation}
	 	where we used that $F = 0$ in the support of $\chi$. Then, using the above  in \eqref{g-t} gives
	 	\begin{align}
	 		\| e^{\bF} g  (H )\| 					\nonumber 
	 		&
	 		\leq 
	 		\frac{1}{\pi } 
	 		\int_{\Delta}
	 		\|(H-z)^{-1}\| \|\bE e^{-\bF}(\widetilde{H}-z^\ast)^{-1}e^{\bF}\|  
	 		| \partial_{\bar z} \,  \tilde{  g  }  (z)| \, \d x \d y  \ , 
	 		\\
	 		& \leq 
	 		\sup_{z\in \Delta}  
	 		\|\bE ( e^{-\bF} 
	 		\widetilde H 
	 		e^{\bF } -z)^{-1}\|
	 		\,\, \,  
	 		\frac{1}{\pi} \!  
	 		\int_\Delta  
	 		\frac{       | \partial_{\bar z }         \, 	 \tilde{  g  }  (z)     |  }{ | \mathrm{Im}z|  }\, \d x \d y  
	 		\leq 
	 	\frac{ 4 \delta_0^{-1 }}{\pi }
	 		\,\,   
	 		\int_\Delta  
	 		\frac{       | \partial_{\bar z }         \, 	 \tilde{  g  }  (z)     |  }{ | \mathrm{Im}z|  }\, \d x \d y\,,  				\label{HS 1}
	 	\end{align}
	 	where in the last step we used 
	 	Lemma~\ref{lemma: HF extension}. Moreover,
	 	the integral above   is finite thanks to the almost-analytic property $ | \partial_{\bar z} \,   \tilde{  g  }   (z)|  = \mathcal{O} (|\mathrm{ Im} \, z  |)$ as $\mathrm{ Im}z \to 0.$   This shows the theorem for $F\in \mathcal{W}_b$.

	 	In order to pass to sequences $F \in  \W   $, we argue     as in \cite{CHSV21}. 
	 	Namely, for each $ n \in \N$ we consider the new sequence  $F_n = (F_{j,n})_{ j \in \Z } $ defined as 
	 	\begin{equation}
	 		F_{j,n} := 
	 		\frac{F_j }{ 1 + n^{-1 } F_j } \, , \qquad  j \in \Z \ , n \in \N \ .
	 	\end{equation}
	 	Clearly, this new sequence remains in $W_{\t{loc}}^{1,2}$. 
	 	In addition the  the uniform bounds $|F_{j,n}| \leq n $ show that the sequence is bounded, i.e. $F_n \in \W_b $ for all $ n \in \N$. 
	 	Furthermore,  the collection of estimates  
	 	\begin{equation}
	 		\label{Fn estimates}
	 		|F_{j,n} ' | \leq |F_{j   }'|     \ , 
	 		\quad \quad 
	 		| F_{j,n} | \leq |F_{j}| 
	 		\quad \t{and} \quad 
	 		|  F_{j,n} - F_{k,n} | \leq |F_j - F_k| 
	 	\end{equation}
	 	show that $F_n$ satisfies the conditions 
\eqref{F1}, 	 	\eqref{F2}, and \eqref{F3},
 with respect to the same constants. 
	 	Hence, we may apply the version of the theorem for     bounded sequences to each member of the collection $ ( F_n)_{n \in \N } $. 
	 	Finally, we note that  since $F_{j,n}  \rightarrow  F_j $  monotonically as $ n \rightarrow \infty $, the Monotone Convergence Theorem implies that 
	 	\begin{equation}
	 		\| e^{\bF} 	\p_I (H)  \| 
	 		=
	 		\lim_{ n \rightarrow \infty }
	 		\|  e^{F_n}   \p_I (H ) \| 
	 		\leq 
	 		\sup_{ n \in \N }
	 		\|  e^{F_n}   \p_I (H ) \| 
	 		< \infty  
	 	\end{equation}
	 	where finitenes of the the last quantity follows from \eqref{HS 1}.  This finishes the proof of the theorem. 
	 \end{proof}


	 \appendix
	 \section{Control of $e^F W e^{-F} $}\label{fourier-F}
	 
	 \begin{lemma}\label{lemma exp twisted}
	 	Let $W: \R^2 \rightarrow \C^2 \times \C^2  $ satisfy Condition \ref{condition 2} with exponent $\beta>0$ 
	 	and potential $v(r)$.
	Then, for all    $F = (F_{j})_{j \in \Z} \in   \W $ satisyfing
	  \eqref{F3}
	  there holds 
	 	\begin{equation}
	 		\|   e^{\bF} W e^{-\bF} \psi \|^2 
	 		\leq
	 		\coth^2(\beta/4 )  \| \,v \psi\|^{2}  \  ,
	 		\qquad \forall \psi \in \mathscr C \  . 
	 	\end{equation} 
	 \end{lemma}
	 
	 \begin{proof}
	 	Let us take smooth functions $\psi= \oplus_{j \in \Z} \psi_j$,  where the decomposition is with respect to the splitting  $( P_{m_j} )_{ j \in \Z  }$.  In order to avoid confusions we make explicit, in a sub-index, the reference space for the norm or the scalar product. 
	 	Let us recall that we denote by $(x,y)$ the inner product on $\C^2$, 
	 	and use the norm $|x|^2  = (x,x)$. 
	 	We write
	 	\begin{align*}
	 		\|v\,\psi\|^2  =\,	\sum_{j \in\Z} \<\psi_j, v^2 \psi_j\>_{L^2(\R^+,\C^2)}=  \sum_{j \in\Z} \int_{0}^\infty v^2(r)
	 		| \psi_j(r)|^2
	 		\, r  \d r =\int_{0}^\infty  v^2(r)\|\psi (r) \|_{\ell^2(\Z)}^2 \, r  \d r\,.
	 	\end{align*} 
	 	A straightforward calculation, shows that for 
	 	$\vp =\oplus_{j \in \Z} \vp_j$ smooth, we have
	 	\begin{equation}\label{lemma 5.4 eq 1}
	 		\begin{split}
	 			|  \<\vp ,      
	 			e^{\bF} W  e^{-\bF}
	 			\psi 
	 			\>_\H  | 
& = 
\bigg|
\sum_{j , k \in \Z } 
\<
P_{m_j} \vp ,  
e^{F_j}
W
e^{-F_k}
P_{m_k}
\psi 
 \>_\H 
\bigg|	 \\ 
&  =  
\bigg|
\sum_{j , k \in \Z } 
\<
 \vp_j ,  
e^{F_j}
\widetilde{W} (\cdot,  j  - k )
e^{-F_k} 
\psi_k 
 \>_{	\widehat \H }
\bigg|		\\ 
	 			& \leq 
	 			\sum_{j,k \in\Z}
	 			e^{F_j(r ) - F_k(r )}
	 			\int_{0}^\infty 
	 			\Big| 
	 			\bigg( \vp_j (r),
	 			\widetilde{W}(r,j-k)
	 			\psi_k(r)
	 			\bigg)    \Big|  
	 			\, r \d r 	\\
	 			&\le 			\int_{0}^\infty 
	 			\sum_{j,k \in\Z}
	 			e^{\beta |j-k|/2} \|\widetilde{W}(r,j-k)\|_{2\times 2}\, 
	 			 | \varphi_j(r) | 
	 			| \psi_k(r)| \, r  \d r \\
	 			&\le 	
	 			\int_{0}^\infty 	v(r)  \,
	 			\Big[	\sum_{j,k \in\Z} 
	 			e^{-\beta |j-k|/2} 
	 			\,
	 			 | \varphi_j(r) | 
	 		| \psi_k(r)|
	 			\Big]
	 			\, r  \d r  \ , 
	 		\end{split}
	 	\end{equation}
where in the last  inequality we   used that
$W$ satisfies Condition \ref{condition 2}.  
	 	Using Young's convolution  inequality we can estimate 
	 	$$
	 	\sum_{j,k \in\Z} e^{-\beta |j-k|/2} 
	 	\,  
	 	| \varphi_j(r) | 
	 	 | \psi_k(r)  | 
	 	\le \sum_{j\in \Z} e^{-\beta |j| / 2} \|
	 	\|\psi(r)\|_{\ell^2}	\|\varphi(r)\|_{\ell^2}=  \coth(\beta/4 ) 
	 		\|\psi(r)\|_{\ell^2}	\|\varphi(r)\|_{\ell^2}\,.
	 	$$
	 	Hence, using Schwarz' inequality in the $r$ integral, we get
	 	\begin{align*}
	 		|  \<\vp ,    \nonumber 
	 		e^{\bF} W  e^{-\bF}
	 		\psi 
	 		\>_\H  | 
	 		& \leq  \coth(\beta/4 ) 
	 		\int_{0}^\infty 	v(r) \, \|\psi(r)\|_{\ell^2}	\|\varphi(r)\|_{\ell^2}
	 		\, r  \d r \,\le \coth(\beta/4 ) \|v\, \psi\| \|\varphi \|\,.
	 	\end{align*}
The proof is finished after we take the supremum over all $\vp$ with unit norm in $\H$. 
	 \end{proof}

	 \section{Domain Cores}\label{appendix domain core}
	 Let $H_0$ be the Dirac operator satisfying Condition \ref{condition 1}, and let $\mathscr{C}$ be the space of functions defined in Eq.  \eqref{domain core}. 
	 Essential self-adjointess of   $H_0 |_{\mathscr{C}}$  follows from essential self-adjointess of $  H_0 |_{C_c^\infty (\R^2 )}$ and  the following result. 
	 Since the proof for $H |_{\mathscr{C}}$ presents no modification, we omit it. 
	 \begin{lemma}
	 	Let $\vp \in C_c^\infty (\R^2 , \C^2)$ and consider the sequence $(\vp_n)_{ n\in \N}$ defined by  
	 	\begin{equation}
	 		\vp_n( \bx )
	 		: 	=
	 		\bigg(
	 		\frac{1}{1 + i n \exp (1 / |\bx |)}
	 		\bigg)
	 		\ 
	 		\vp (\bx )
	 		\, , 
	 		\qquad 
	 		\bx \in \R^2 \backslash \{0\} \, , 
	 	\end{equation}
	 	and $\vp_n(0)=0$. Then, $\vp_n \in \mathscr{C}$ for every $ n \in \N$  and 
	 	\begin{equation}
	 		\lim_{ n \rightarrow \infty } 
	 		\Big( 
	 		\| H_0 (\vp - \vp_n)  \| 
	 		+
	 		\|  \vp - \vp_n \|
	 		\Big) 
	 		=0  \ . 
	 	\end{equation}
	 \end{lemma}

	 \begin{proof}
	 	Given $ n \in \N$, for the rest of the proof  we let $\chi_n(r) =  \big(1 + i n \exp (1/r) \big)^{-1 }$ for $ r >  0 $. 
	 	In particular, it is straightforward--although tedious--to verify that $\chi_n \in C^\infty (0, \infty)$ and that it satisfies the following asymptotic estimates as $r \downarrow 0 $
	 	\begin{equation}
	 		|  \chi^{(k)} (r)| = \mathcal{O} (r^N) \ , \qquad  \forall  N \in \N \ ,   \forall k \in \N_0  \ .
	 	\end{equation}
	 	Since $\vp $ and all of its derivatives are bounded, it follows easily that $\vp_n \in \mathscr{C}$. For the second part, we use the triangle inequality and a commutator estimate to find that
	 	\begin{align}
	 		\|   H_0 (\vp - \vp_n)  \|    
	 		=
	 		\|  H_0 (1 - \chi_n ) \vp   \|
	 		\leq 
	 		\|   (1 - \chi_n) H_0 \vp   \|
	 		+
	 		\| \chi_n' \vp    \| \ .
	 	\end{align}
	 	Since for all $r>0$, $\lim_{ n \rightarrow \infty } \chi_n(r) = 1 $ and $|1 - \chi_n| \leq 2 $, the Dominated Convergence Theorem implies that $\|   (1 - \chi_n) H_0 \vp   \|$ converges to zero. 
	 	For the second term, we calculate that for all $r > 0 $
	 	\begin{equation}
	 		| \chi_n ' (r) | 
	 		\ = \ 
	 		\bigg| 
	 		\frac{ n \exp( 1 /r ) r^{-2 }    }{  \big(  1 + i n \exp (1/r) \big)^2   }
	 		\bigg| 
	 		\  = \ 
	 		\frac{  n  e^{-1/r} r^{-2 } }{     e^{-2/r }		+ n^2 }
	 		\ \leq  \ 
	 		\frac{ \sup_{ \rho >0}
	 			\big(
	 			e^{-1/ \rho } \rho^{-2 }
	 			\big) }{n} \ .
	 	\end{equation}
	 	Thus, we put the two estimates together to find that $ \lim_{ n \rightarrow \infty } \|  H_0 (\vp_n - \vp ) \| = 0$. The estimate for $ \| \vp - \vp_n \|$ only uses the Dominated Convergence Theorem. This finishes the proof of the lemma.   
	 \end{proof}
	 
	 Let  now fix $ j \in \Z$ and let $h_j$ be the operator defined in Eq. \eqref{Dj definition}, acting on $L^2 (  \R_+  , \C^2     )$. 
	 Recall that the space of radial functions $\Cr$ is defined in Eq. \eqref{domain core C'}. 
	 
	 \begin{lemma}
	 	\label{lemma domain core}
	 	$\Cr$ is a domain core of  	$h_j$.
	 \end{lemma}

	 \begin{proof}
	 	Consider the canonical projection $\Pi_j : \bigoplus_{ \ell \in \Z} L^2 (\R_+, \C^2 ) \rightarrow L^2(\R_+, \C^2 )$ 
	 	into the $j$-th component. 
	 	Since $H_0$
	 	is essentially self-adjoint in $\mathscr C$, 
	 	its $j$-th block
	 	$h_j$
	 	is essentially self-adjoint in $\Pi_j \calF \calU \mathscr{C}$.  
	 	We claim that 
	 	\begin{equation}
	 		\Pi_j \calF \calU
	 		\mathscr{C} \subset \Cr
	 	\end{equation}
	 	from which it is easy to conclude that $\Cr$ is a domain core of $h_j$. 
	 	To this end,  let $\vp  = (\vp_1 , \vp_2)^\top \in   \Pi_j \calF \calU$. 
	 	Then, using the definition of the maps $\calF$ and $\U$, it verifies that there exists $\psi = ( \psi_1 , \psi_2)^\top  \in \mathscr{C}$ such that the following formula holds 
	 	\begin{equation}\label{fourier  eq 1}
	 		\vp_1 (r) = (2 \pi)^{-1/2}
	 		\int_0^{2\pi }
	 		\psi_1 (r \cos \theta , r \sin \theta) e^{ - \im j \theta } d \theta  \, , \qquad r>0 
	 	\end{equation}
	 	and analogously  for $\vp_2$. 
	 	Starting from the formula in Eq. 
	 	\eqref{fourier  eq 1}, it is staightforward to verify that $\vp \in \Cr $. This finishes the proof   
	 \end{proof}

	 Let us now turn to the final result in this appendix.  
	 Let $\widetilde{H}_0$  be  the operator defined as    in \eqref{h0-tilde definition}, 
	 in the space 
	 $ \calC = \{   \vp \in \bigoplus_{j\in \Z} \Cr   \ | \ \exists N : \vp_j = 0 \t{ if } |j| >   N  \} $. 
	 
	 \begin{lemma}\label{lemma-essential}
	 	$\widetilde{H}_0$ is essentially self-adjoint. 
	 \end{lemma}
	 
	 \begin{proof}
	 	Clearly, $\widetilde{H}_0$ is symmetric.  
	 	We will show that $ ( \F \U )\t{Ran}(  \widetilde{H}_0 \pm \im 		\textbf{1}  )$ is dense in 
	 	$ 	\widehat \H = \bigoplus_{j\in \Z} L^2(\R_+ , \C^2 ) $. 
	 	To this end, let $ \psi  = (\psi_j)_{j\in \Z } \in 
	 	 	\widehat \H$    and pick $\epsilon >0$. Choose $ N >0$ such that 
	 	\begin{equation}
	 		\|   \psi  - \psi^{(N)}  \|_{    	\widehat \H }
	 		\leq \epsilon /2
	 	\end{equation}
	 	where $\psi^{(N)}_j$ is defined as $\psi_j$ for $| j| \leq N$ and as $0$ for $|j| > N$. 
	 	Let us fix $ j \in \{ 1 , \ldots, N\}$ and let us note that   $ h_j |_{\Cr}$ is essentially self-adjoint.
	 	Since $\chi_j$ is     bounded, it is easy to check that the same statement holds for $ \widetilde h_j  = h_j + \sigma_3 E_j \chi_j$. 
	 	In particular, it follows   that  $\mathrm{Ran}(  \widetilde{h}_j |_{\Cr}   \pm \im \1       )$ is dense in $ L^2 (\R_+, \C^2 )$. 
	 	Therefore, we may find $\vp_j \in    \Cr $ such that 
	 	\begin{equation}
	 		\|   ( \widetilde h _j   \pm   \im 		\textbf{1}  )\vp_j - \psi_j  \| \leq \epsilon / (2N). 
	 	\end{equation}
	 	Let us consider $ \vp^{(N)} = (\vp_j)_{j\in \Z} \in \calC $, where we set $\vp_j \equiv  0$ for $| j  |> N$. 
	 	The triangle inequality then implies that 
	 	\begin{equation}
	 		\|   (\F \U ) (\widetilde H _0   + \im ) (\F \U)^{-1 }\vp^{(N) }   
	 		- \psi 
	 		\|
	 		\leq 
	 		\|    (\F \U ) (\widetilde H _0   + \im 		\textbf{1} ) (\F \U)^{-1 }  \vp^{(N) }   
	 		- \psi^{(N)}
	 		\|
	 		+
	 		\|  \psi - \psi^{(N)}  \| \leq \epsilon \ . 
	 	\end{equation}
	 	Since $\psi  \in \tilde \H $ and $\epsilon >0 $ were arbitrary, we conclude that   $\t{Ran}(  \widetilde{H}_0 \pm \im	\textbf{1}   )$ is dense in  $ 	\widehat \H$; essential self-adjointess now follows from standard arguments. 
	 \end{proof}

	 \section{Local Integrability}
	 \label{appendix local}
	In what follows we write $\1_R (\bx)\equiv \1_{B_r(0)} (\bx)\,\, (\bx\in \R^2)$
for
the indicator function on the centered ball of radius $r>0$. 	 
	 \begin{lemma}\label{loc.int}
	 	Consider the Dirac operator 
	 	$H_0 = \sigma \cdot (- i \nabla -A )$  
	 	with a locally integrable  magnetic field $B$.
	 	Let $V \in L^2_{\rm loc} (\R^2  )$. 
	 	Then, for all $R>0$
	 	and all $\ve>0$, 
	 	there exists $C = C(R, \ve)>0$
	 	such that for all $\vp \in C_0^\infty(\R^2)$
	 	there holds 
	 	\begin{align}
	 		\| V \1_R \vp \| 
	 		\leq 
	 		\ve 
	 		\|  H_0 \vp  \|
	 		+ 
	 		C \|  \vp   \| \ . 
	 	\end{align}
	 \end{lemma}
	 \begin{remark}\label{rem.ess.sa}
	 	Consider $W$ satisfying Condition~\ref{condition 2} and  
	 	notice that thanks to Lemma~\ref{lemma exp twisted} (with ${\bf F}=0$) and Lemma~\ref{loc.int} we get that $W\1_R$ is infinitesimally small with respect to $H_0$. In particular, $H_0+W\1_R$ is essentially self-adjoint on $C_c^\infty(\R^2,\C^2)$ for every $R>0$. Hence by the theorem of Chernoff (see \cite{Chernoff73,Thaller1992}),  $H=H_0+W$  is essentially self-adjoint on $C_c^\infty(\R^2,\C^2)$.
	 	\end{remark}
	 \begin{proof}
	 	For fixed $R>0$ consider $ 0 \leq \chi_R \leq 1$
	 	a smooth function of compact support 
	 	satisfying 
	 	$\1_R \leq \chi_R^2$.
	   It is well-known that a multiplication operator is 
	   infinitesimally form bounded by $(-\Delta)$ if 
	   it  belongs to $L^1(\R^2) + L^\infty ( \R^2)$
	   (for a proof use Sobolev Inequality). Since $V^2 \1_R$ belongs to that  space we get, using the Diamagnetic Inequality that, for all $\ve>0$ and smooth $\vp$,
	 	\begin{align}
	 		\|     V \1_R \vp    \|^2 
	 		\leq 
	 		\< \vp_R , V^2 \1_R  \vp_R  \>  
	 		\leq 
	 		\ve \< \vp_R , ( - i \nabla -A)^2  \vp_R\>
	 		+ 
	 		C_\ve \|  \vp_R \|^2
	 	\end{align}
	 	for some $C_\ve>0$, 
	 	and we denote 
	 	$\vp_R \equiv \chi_R \vp$. 
	 	It is well-known that
	 	$H_0^2 = (- i \nabla -A)^2 + \sigma_3 B$
	 	and, consequently, 
	 	a simple computation shows that 
	 	\begin{align}
	 		\nonumber 
	 		\|     V \1_R \vp    \|^2 
	 		&  \leq 
	 		\ve 
	 		\|      H_0 \vp_R  \|^2 
	 		- 
	 		\ve 
	 		\<\vp_R , \sigma_3 B \vp_R  \>
	 		+ 
	 		C_\ve \|  \vp_R \|^2   \\ 
	 		& \leq 
	 		\ve 
	 		\|     H_0 \vp    \|^2
	 		+
	 		\ve 
	 		\< \vp ,
	 		\chi_R  | B|   \chi_R 
	 		\vp\> 
	 		+ 
	 		\Big( 
	 		\ve   \|   \nabla \chi_R   \|_{L^\infty}^2
	 		+ 
	 		C_\ve
	 		\Big) \|  \vp \|^2  \ . 
	 		\label{c4}
	 	\end{align}
	 	Since $B$ is locally integrable, 
	 	there holds
	 	$\chi_R |B|^{1/2}\in L^2(\R^2)$.
	 	Thus, repeating the same argument 
	 	we used to control $V$, 
	 	we find 
	 	\begin{align}
	 		\< \vp ,
	 		\chi_R  | B|   \chi_R 
	 		\vp\>  
	 		\leq 
	 		\ve 
	 		\|     H_0 \vp    \|^2
	 		+
	 		\ve 
	 		\< \vp ,
	 		\chi_R  | B|   \chi_R 
	 		\vp\> 
	 		+ 
	 		\Big( 
	 		\ve   \|   \nabla \chi_R   \|_{L^\infty}^2
	 		+ 
	 		C_\ve
	 		\Big) \|  \vp \|^2  \ 
	 	\end{align}
	 	for a possibly different constant $C_\ve$.
	 	One may  use the above inequality 
	 	to conclude that 
	 	$ \chi_R |B| \chi_R$
	 	is  form bounded by $H_0^2$, 
	 	with relative constant $\ve/ (1-\ve)$. 
	 	This fact, 
	 	combined with the bound \eqref{c4} finishes the proof. 
	 \end{proof}

	 \section{Reduction to Condition \ref{condition 1'}}
	 \label{section scaling}
	 In Condition \ref{condition 1} we consider the magnetic flux 
	 $\Phi(r )  = \Phi_0 r^{\alpha + 1 } + \phi(r) $ 
	 with $\alpha$, $\Phi_0>0$ and $\phi(r) = o(r^{1+ \alpha})$ as $r \rightarrow \infty$.
In this Appendix, we argue that it suffices to prove Theorem \ref{theorem tunneling estimates} 
for the magnetic field that satisfies Condition \ref{condition 1'}.
 That is, the case $A(r ) = r^\alpha$. First we
discuss how to remove $\phi$ and then, by a scaling argument, we check how to remove the factor in front of $A$.

 \subsection{The magnetic perturbation}
Assume that we haven proven Theorem~\ref{theorem tunneling estimates} for the magnetic flux  $\Phi(r)=\Phi_0r^{\alpha+1}$, 
and all $W$ satisfying Condition \ref{condition 2}. 
Then, we argue that the theorem holds also when such 
$\Phi(r)$ is replaced by 		
 \begin{align}
 \label{magnetic field2}
 {\Phi}^\prime ( r ) 
 \, =  \, 
 \Phi_0 \, 
 r^{ \alpha +1 } 
 + \phi(r) \ . 
 \end{align}	
 Indeed, let us write $H_0$ and $H_0^\prime$  for the Hamiltonians whose fluxes 
 are $\Phi(r)$ and $\Phi'(r)$, respectively.
 Then, we write 
 $$
 A^\prime(r)={\Phi}^\prime ( r ) /(2\pi r)=A(r)+a(r)\,,
 $$
 where $A(r)={\Phi} ( r ) /(2\pi r)$ and $a(r)=o(r^\alpha)$. Hence,  we find that (see \eqref{aaa}) $H_0^\prime=H_0+ a(r)({\boldsymbol{\sigma}\cdot {\textbf{e}_\theta }})$.
 Therefore, for $W$ satisfying Condition~\ref{condition 2} we have 
 $$
 H'=H_0^\prime+W=H_0+W+
 a(r)({\boldsymbol{\sigma}\cdot {\textbf{e}_\theta }})
 \,.
 $$
Finally, note that $W+a(r)({\boldsymbol{\sigma}\cdot {\textbf{e}_\theta }})$ satisfies Condition~\ref{condition 2}. 
 In fact, one can explicitly verify using \eqref{F unitary} that $a(r)({\boldsymbol{\sigma}\cdot {\textbf{e}_\theta }})$ is rotationally symmetric.   
  Hence, Theorem~\ref{theorem tunneling estimates} holds for $H'$.

 \subsection{Scaling}
	  In this subsection, 
	 we argue that using a scaling transformation we
	 can always reduce the analysis  to the notationally simpler case $A(r) = r^\alpha$.
	 We assume here that we start from $\Phi(r) = \Phi_0	r^{\alpha +1  }$. 
	 
	 \vspace{2mm}
	 To this end, let $\lambda>0$ and define the unitary transformation 
	 \begin{equation}
	 	S_\lambda : L^2(\R^2; \C^2) \rightarrow L^2(\R^2; \C^2) 
	 	\ , \qquad  ( S_\lambda \vp )(x) \equiv  \lambda^{ -2 } \vp( \lambda^{-1} x )  \ , 
	 \end{equation}
	 which implements the scaling $x \mapsto \lambda^{-1 } x $. 
	 In particular, we leave it 
	 as an exercise to the reader to verify that 
	 the conjugation of the Hamiltonian $H$
	 with the scaling transformation  is given by 
	 \begin{equation}
	 	\label{hamiltonian scaled}
	 	S^*_\lambda  H S_\lambda 
	 	= \frac{1}{\lambda}
	 	H^{(\lambda)}
	 	\qquad
	 	\t{where}
	 	\qquad 
	 	H^{(\lambda)}
	 	\equiv 
	 	{\sg   } \cdot (-\im \nabla - {\bf A}_\lambda ) 
	 	+ 
	 	W_\lambda 
	 \end{equation}
	 is the re-scaled Hamiltonian
	 with new potentials  
	 $\bf{A}_\lambda  ( \bx ) \equiv  \lambda \bf{A} (\lambda \bx )$
	 and 
	 $W_\lambda( \bx) \equiv \lambda W (\lambda \bx )$. 
	 
	 \begin{remark}
	 	Let us observe that the following facts are true. 
	 	
	 	\begin{enumerate}[leftmargin=*]
	 		\item The re-scaled vector potential is given by 
	 		$ {\bf{A}}_\lambda ( \bx ) 
	 		=  A_\lambda ( |\bx |) {\textbf{e}_\theta }$
	 		with 
	 		$
	 		A_\lambda (r) \equiv  	  \lambda^{1+ \alpha }  \frac{\Phi_0}{2\pi } r^{\alpha}
	 		$.  
	 		
	 		\item The re-scaled  perturbation  $W_\lambda $
	 		satisfies Condition \ref{condition 2}
	 		with respect to the same constant $\beta>0$, 
	 		the re-scaled   radial function $v_\lambda (r) = \lambda v(\lambda r )$, 
	 		and the same parameter $\mu_*$.
	 		This follows from  
	 		\begin{equation}
	 			\mu_*  = 
	 			\limsup_{r \to \infty}
	 			\bigg|
	 			\frac{\coth^2 (\beta/4)	 v_\lambda^2(r)	}{ A_\lambda^2 (r)}
	 			\bigg| \ .  
	 		\end{equation}
	 		
	 		\item 
	 		For all measurable functions $f,g : \R^2 \rightarrow \C $
	 		the following relations hold 
	 		for the operators $J$ and $\bx $ 
	 		\begin{equation}
	 			S_\lambda^*  f(J)  S_\lambda = f(J)
	 			\qquad 
	 			\t{and}
	 			\qquad 
	 			S_\lambda^* g( \bx )	S_\lambda = g (	 \lambda^{-1 } \bx 	)
	 		\end{equation}
	 	\end{enumerate}
	 	
	 \end{remark}

	 Let $\lambda>0$ be fixed. 
	 In view of the above observations, 
	 we argue  that if one re-scales
	 the potentials  
	 $W \mapsto W_\lambda$ and $A \mapsto A_\lambda$, 
	 then 
	 Theorem \ref{theorem tunneling estimates} holds
	 for the re-scaled Hamiltonian, 
	 and this operation  only changes the value of the constants $\zeta_1$, $\zeta_2$, $C_1$
	 and $C_2$, 
	 and the energy treshold $E>0$. 
	 In practice, we   choose $\lambda$
	 so that
	 \begin{equation}
	 	\label{A simple}
	 	A_\lambda (r)	=		r^{\alpha } , \qquad r>0 
	 \end{equation}
	 which will simplify several calculations.
	 Explicitly $\lambda = (\frac{2\pi}{\Phi_0}	)^{\sigma}$.

	 \vspace{2mm}	 
	 
	 More precisely,  
	 let us consider $E>0$.
	 Denote by   $ I  = [ - E, E]$   an energy interval, 
	 and  let  $I (\lambda ) = [ - \lambda E , \lambda E ]$
	 be its re-scaled version. 
	 Then, the following result is true
	 and its proof is left to the reader. 
	 
	 \begin{lemma}
	 	For all    measurable functions $f,g : \R^2 \rightarrow \C$
	 	and $\lambda>0$
	 	there holds 
	 	\begin{equation}
	 		\label{scaling}
	 		\|	  f(J) g(\bx ) \mathbb{P}_{I}(H )  	\| 
	 		= 
	 		\|			 f( J )	 g(\lambda^{-1 }{ \bx }	 )   \mathbb{P}_{  I	(\lambda) } (   H^{(\lambda) } )		\|
	 	\end{equation}
	 	where $H^{(\lambda)}
	 	$
	 	is the re-scaled Hamiltonian \eqref{hamiltonian scaled}. 
	 	In particular, for all $C_1$, $C_2$,  $\zeta_1$,  $\zeta_2>0$
	 	and all $ j \in \Z $
	 	there holds
	 	\begin{align}
	 		\nonumber
	 		\|
	 		\exp \big(    \zeta_1    |j | \big) 
	 		\, 
	 		\1_{ [ 0, C_1 |j|^\sigma]   }(r )
	 		P_{m_j}
	 		\p_I( H )
	 		\| 
	 		& 	 = 
	 		\|
	 		\exp \big(    \zeta_1^{(\lambda )}    |j | \big) 
	 		\, 
	 		\1_{ [ 0, C_1^{(\lambda)}  |j|^\sigma]   }(r )
	 		P_{m_j}
	 		\p_{I(\lambda) }( H^{(\lambda)} )
	 		\|  	\\ 
	 		\nonumber
	 		\|
	 		\exp \big(   \zeta_2  \,   r^{ 1 + \alpha }  \big) 
	 		\1_{   [ C_2 \<j \>^\sigma  , \infty ]}(r ) 
	 		P_{m_j}
	 		\p_I (   H ) 
	 		\| 
	 		&	 =
	 		\|
	 		\exp \big(   \zeta_2^{(\lambda)}  \,   r^{ 1 + \alpha }  \big) 
	 		\1_{   [ C_2^{(\lambda)} \<j \>^\sigma  , \infty ]}(r ) 
	 		P_{m_j}
	 		\p_{ I (\lambda)} (   H^{(\lambda)} ) 
	 		\| \  ,	\nonumber
	 	\end{align}
	 	where $\zeta_1^{ (\lambda)}  = \zeta_1 , \  $
	 	$\zeta_2^{ (\lambda)} = \lambda^{-1 - \alpha}  \zeta_2  , \  $
	 	$C_2^{(\lambda)} = \lambda C_2, \  $
	 	and
	 	$ \ C_1^{(\lambda)} = \lambda C_1   $
	 	are the re-scaled constants. 
	 \end{lemma}

	 \bigskip
	 \noindent {\bf Acknowledgments.}
E.C. gratefully acknowledges support from the Provost’s Graduate Excellence Fellowship at The University of Texas at Austin and from the NSF grant DMS-2009549, and the NSF grant DMS-2009800 through
 Thomas Chen.
	B.P. and E.S  acknowledge  support from  Fondecyt (ANID, Chile) 
	through the  grant \# 123--1539. E.S. thanks Marcel Griesemer for an observation concerning  
	 Remark~\ref{rem.grie}.

	 \vspace{1mm}
	 \noindent \textbf{Conflict of interest.}
	 The authors declare that there is no conflict of interest.

 \end{document}